\documentclass[12pt, paper=a4,fleqn]{scrartcl}
\usepackage[utf8]{inputenc}

\usepackage[utf8]{inputenc}

\usepackage{amsmath}
\usepackage{amssymb}
\usepackage{mathtools}

\mathtoolsset{centercolon}

\DeclarePairedDelimiter\paren{\lparen}{\rparen}
\DeclarePairedDelimiter\angles{\langle}{\rangle}
\DeclarePairedDelimiter\braces{\{}{\}}
\DeclarePairedDelimiter\brackets{[}{]}
\DeclarePairedDelimiter\abs{\lvert}{\rvert}
\DeclarePairedDelimiter\norm{\lVert}{\rVert}

\DeclarePairedDelimiterX{\closedStochasticInterval}[1]{[}{]}{\!\delimsize[#1\delimsize]\!}
\DeclarePairedDelimiterX{\leftOpenStochasticInterval}[1]{]}{]}{\!\delimsize]#1\delimsize]\!}
\DeclarePairedDelimiterX{\rightOpenStochasticInterval}[1]{[}{[}{\!\delimsize[#1\delimsize[\!}
\DeclarePairedDelimiterX{\openStochasticInterval}[1]{]}{[}{\!\delimsize]#1\delimsize[\!}

\newcommand{\bp}[1]{\paren[\big]{#1}}
\newcommand{\maxET}{{\eta_{\max}}}

\newcommand{\diff}{\,\mathrm d}

\newcommand{\sgn}{\operatorname{sgn}}

\newcommand{\indicator}{\mathds{1}}
\newcommand{\EE}{\mathds{E}}
\newcommand{\NN}{\mathds{N}}
\newcommand{\PP}{\mathds{P}}
\newcommand{\QQ}{\mathds{Q}}
\newcommand{\RR}{\mathds{R}}

\newcommand{\scA}{\mathcal{A}}

\newcommand{\scE}{\mathcal{E}}
\newcommand{\scF}{\mathcal{F}}
\newcommand{\scH}{\mathcal{H}}

\newcommand{\scX}{\mathcal{X}}

\newcommand{\squeeze}[2][0]{%
  \mbox{$\medmuskip=#1mu\displaystyle#2$}%
}

\newcommand{\baseS}{\widebar{S}}
\newcommand{\admissibleSellStrategies}[1]{\scA_\text{mon}(#1)}
\newcommand{\admissibleFiniteVariationStrategies}[1]{\scA_\text{bv}(#1)}
\newcommand{\admissibleSemimartingaleStrategies}{\scA_\text{semi}}

\newcommand{\assetsProcess}{\Theta}

\newcommand{\marcusDriver}{\Phi}

\newcommand{\impactVolatility}{\hat\sigma}

\makeatletter
\newcommand*\if@single[3]{%
  \setbox0\hbox{${\mathaccent"0362{#1}}^H$}%
  \setbox2\hbox{${\mathaccent"0362{\kern0pt#1}}^H$}%
  \ifdim\ht0=\ht2 #3\else #2\fi
  }
\newcommand*\rel@kern[1]{\kern#1\dimexpr\macc@kerna}
\newcommand*\widebar[1]{\@ifnextchar^{{\wide@bar{#1}{0}}}{\wide@bar{#1}{1}}}
\newcommand*\wide@bar[2]{\if@single{#1}{\wide@bar@{#1}{#2}{1}}{\wide@bar@{#1}{#2}{2}}}
\newcommand*\wide@bar@[3]{%
  \begingroup
  \def\mathaccent##1##2{%
    \if#32 \let\macc@nucleus\first@char \fi
    \setbox\z@\hbox{$\macc@style{\macc@nucleus}_{}$}%
    \setbox\tw@\hbox{$\macc@style{\macc@nucleus}{}_{}$}%
    \dimen@\wd\tw@
    \advance\dimen@-\wd\z@
    \divide\dimen@ 3
    \@tempdima\wd\tw@
    \advance\@tempdima-\scriptspace
    \divide\@tempdima 10
    \advance\dimen@-\@tempdima
    \ifdim\dimen@>\z@ \dimen@0pt\fi
    \rel@kern{0.6}\kern-\dimen@
    \if#31
      \overline{\rel@kern{-0.6}\kern\dimen@\macc@nucleus\rel@kern{0.4}\kern\dimen@}%
      \advance\dimen@0.4\dimexpr\macc@kerna
      \let\final@kern#2%
      \ifdim\dimen@<\z@ \let\final@kern1\fi
      \if\final@kern1 \kern-\dimen@\fi
    \else
      \overline{\rel@kern{-0.6}\kern\dimen@#1}%
    \fi
  }%
  \macc@depth\@ne
  \let\math@bgroup\@empty \let\math@egroup\macc@set@skewchar
  \mathsurround\z@ \frozen@everymath{\mathgroup\macc@group\relax}%
  \macc@set@skewchar\relax
  \let\mathaccentV\macc@nested@a
  \if#31
    \macc@nested@a\relax111{#1}%
  \else
    \def\gobble@till@marker##1\endmarker{}%
    \futurelet\first@char\gobble@till@marker#1\endmarker
    \ifcat\noexpand\first@char A\else
      \def\first@char{}%
    \fi
    \macc@nested@a\relax111{\first@char}%
  \fi
  \endgroup
}
\makeatother

\usepackage{amsthm}
\usepackage{calc}
\usepackage{caption}
\usepackage{subcaption}
\usepackage{color}
\usepackage{csquotes}
\usepackage{dsfont}
\usepackage{graphicx}
\usepackage{import}
\usepackage{microtype}
\usepackage{nameref}
\usepackage[section]{placeins}
\usepackage{tikz}
\usepackage{units}

\usepackage{xr}
\usepackage[unicode,pdfborder={0 0 0}]{hyperref}
\usepackage[noabbrev,capitalize]{cleveref}

\usetikzlibrary{calc}

\newtheoremstyle{boldremark}
	{\topsep}   
	{\topsep}   
	{}          
	{}          
	{\bfseries} 
	{.}         
	{.5em}      
	{}          
	
\newtheorem{theorem}{Theorem}[section]
\newtheorem{proposition}[theorem]{Proposition} 
\newtheorem{lemma}[theorem]{Lemma} 
\newtheorem{corollary}[theorem]{Corollary} 

\theoremstyle{definition}
\newtheorem*{definition}{Definition} 
\newtheorem{assumption}[theorem]{Assumption}

\theoremstyle{boldremark}
\newtheorem{remark}[theorem]{Remark}
\newtheorem{example}[theorem]{Example}

\AtBeginDocument{%
	\crefname{equation}{equation}{equations}%
}

\numberwithin{equation}{section}

\author{Dirk Becherer%
	\footnote{Email addresses: becherer,bilarev,frentrup@math.hu-berlin.de}
 , Todor Bilarev%
	\footnote{
Support by German Science foundation DFG, via Berlin Mathematical School BMS and research training group RTG1845 StoA  is acknowledged.
\newline
We thank Kai Kümmel for an inspiring discussion and the two anonymous referees and associate editor for their helpful and encouraging feedback.}
 , Peter Frentrup
\\
Institute of Mathematics, Humboldt-Universität zu Berlin\\
Unter den Linden 6 - 10099 Berlin
}

\title{Stability for gains from large investors' strategies in $\mathbf{M_1}$/$\mathbf{J_1}$ topologies}

\begin{document}

\maketitle

\begin{abstract}
We prove continuity of a controlled SDE solution in Skorokhod's $M_1$ and $J_1$ topologies and also uniformly, in probability,
 as a non-linear functional of the control strategy.
The functional comes from a finance problem to model price impact of a large investor in an illiquid market.
We show that $M_1$-continuity is the key to ensure that
 proceeds and wealth processes from (self-financing) c\`{a}dl\`{a}g trading strategies 
are determined as the continuous extensions for those from continuous strategies.
We demonstrate by examples how continuity properties are useful to solve different stochastic control problems on optimal liquidation
 and to identify asymptotically realizable proceeds.

\vspace*{2ex}

{{ MSC2010 subject classifications}: }
60H10, 60H20, 60G17, 91G99, 93E20 

\end{abstract}

\section{Introduction}

A classical theme in the theory of stochastic differential equations is how stably the solution process behaves, as a functional of its integrand and integrator processes,
see e.g.\ \cite{KurtzProtter96} and \cite[Chapter~V.4]{Protter04}.
A typical question is how to  extend such a functional sensibly to a larger class of input processes.
Continuity is a key property to address such problems, cf.\ e.g.\ the canonical extension of Stratonovich SDEs by Marcus \cite{Marcus81}.

In singular control problems for instance, the non-linear objective functional may initially be only defined for finite variation or even absolutely continuous control strategies.
Existence of an optimizer might require a continuous extension of the functional to a more general class of controls, e.g.\ semimartingale controls for the problem of hedging.
Herein the question of  which topology to embrace arises, and this depends on the problem at hand, see e.g.\ \cite{Kardaras2013} for an example of utility maximization in a frictionless financial market where the Emery topology turns out to be useful for the existence of an optimal wealth process.
For our application we need suitable topologies on the Skorokhod space of c\`adl\`ag functions. The two most common choices here are the uniform topology and Skorokhod $J_1$ topology; they share the property that a jump in a limiting process can only be approximated by jumps of comparable size at the same time or, respectively, at nearby times.
But this can be overly restrictive for such applications, as we have in mind, where a jump may be approximated sensibly by many small jumps in fast succession or by continuous processes such as Wong-Zakai-type approximations.
The $M_1$ topology by  Skorokhod \cite{Skorokhod56} captures such approximations of \emph{unmatched jumps}. We will take this as a starting point 
 to identify the relevant non-linear objective functional for c\`{a}dl\`{a}g controls as a continuous extension from (absolutely) continuous controls.
See \cite{Whitt2002_book} for a profound survey on the $M_1$ topology. 

We demonstrate how the old subject of stability of SDEs with jumps, when considered with respect to the $M_1$ topology,  has applications for recent problems in mathematical finance.
Our application context is that of an illiquid financial market for trading a single risky asset.
A large investor's trading causes transient price impact on some exogenously given fundamental price which would prevail in a frictionless market.
Such could be seen as a non-linear (non-proportional) transaction cost with intertemporal impact also on subsequent prices.
Our framework is rather general. It can accommodate for instance for models where price impact is basically additive, see \cref{ex: additive or multiplicative impact}; Yet, some 
 extra care is
 required here to ensure $M_1$ continuity, which 
  can actually fail to hold in common additive models that lack a  monotonicity property and positivity of prices,  cf.\ \cref{rmk:positive prices}.
An original aspect of our framework is that it also permits for multiplicative impact
which 
 appears to fit better to 
multiplicative price evolutions  as e.g.\ in models of Black-Scholes type, cf.\ \cite[Example~5.4]{BechererBilarevFrentrup2016-deterministic-liquidation}; 
In comparison, it moreover ensures positivity of asset prices, which is desirable from a theoretical point of view,  relevant for applications whose time horizon is not short (as they can occur e.g.\ for large institutional trades \cite{ChanL95,MaugEtal}, or for hedging problems with longer maturities).

The large trader's feedback effect on prices causes the proceeds (negative expenses) to be a non-linear functional of her control strategy for dynamic trading in risky assets.
Having specified the evolution for an
affected
price process at which trading of infinitesimal quantities would occur, one still has,  even for a simple block trade, to define the variations in the bank account by which the trades in risky assets are financed, i.e.\ the so-called  self-financing condition.
Choosing a seemingly sensible, but ad-hoc, definition could lead to surprising and undesirable consequences, in that the large investor can evade her liquidity costs entirely by using continuous finite variation strategies to approximate her target control strategy, cf.\ \cref{ex:ad-hoc definition of proceeds}.
Optimal trade execution proceeds or superreplication prices  may be only approximately attainable in such models. Indeed,  the analysis in  \cite{BankBaum04,CetinJarrowProtter04} shows that
approximations by continuous strategies of finite variation play a particular role.
This is, of course, a familar theme in stochastic analysis, at least since Wong and Zakai \cite{WongZakai65}. However,
 in the  models in \cite{BankBaum04,CetinJarrowProtter04}  the aforementioned strategies have zero liquidity costs, permitting the large trader to avoid those costs entirely by simply approximating more general strategies. This appears not desirable from an application point of view, and it seems also mathematically inconvenient  to distinguish between proceeds and asymptotically realizable proceeds.
To settle this issue, a stability analysis for proceeds for a class of price impact
models should address in particular the 
$M_1$ topology, in which continuous finite variation strategies are dense in the space of c\`{a}dl\`{a}g strategies (in contrast to the uniform or $J_1$ topologies), see \cref{rmk:M1 important}.

We contribute a systematic study on stability of the proceeds functional.
Starting with an unambiguous definition~\eqref{eq:proceeds cont proc} for continuous finite-variation strategies, 
we identify the approximately realizable gains for a large set of controls. A mathematical challenge for stability of the stochastic integral functional is that both the integrand and the integrator depend on the control strategy.
Our main \cref{thm:stability in j1 and m1} shows continuity of this non-linear controlled functional in the uniform, $J_1$ and $M_1$ topologies, in probability, on the space of (predictable) semimartingale or c\`{a}dl\`{a}g strategies which are bounded in probability.
A particular consequence is a Wong-Zakai-type approximation result, that could alternatively  be shown by adapting results from \cite{KurtzPardouxProtter95} on the Marcus canonical equation to our setup, cf.~\cref{sect:stability}.  
Another direct implication of $M_1$ continuity is that proceeds of general (optimal) strategies can be approximated by those of simple strategies with only small jumps.
Whereas the former property is typical for common stochastic integrals, it is far from obvious for our non-linear controlled SDE functional \eqref{eq:def of proceeds process}.

The topic of stability for the stochastic process of proceeds from dynamically trading risky assets in illiquid markets, where the dynamics of the wealth and of the proceeds for a large trader are non-linear in her strategies because of her market impact, is showing up at several places in the literature.
But the mathematical topic appears to have been touched mostly in-passing so far.
The focus of few notable investigations has been on the application context and on different topologies,
see e.g.\ \cite[Prop.~6.2]{RochSoner13} for uniform convergence in probability (ucp).
In \cite[Lem.~2.5]{LorenzSchied13} a cost functional is extended from simple strategies to semimartingales via convergence in ucp.
\cite[Def.~2.1]{Roch11} and \cite[Sect.~A.2]{CetinJarrowProtter04} use particular choices of approximating sequences to extend their definition of self-financing trading strategies from simple processes to semimartingales by limits in ucp. 
Trading gains of semimartingale strategies are defined in \cite[Prop.~1.1--1.2]{BouchardLoeperZou16} as  $L^2$-limits of gains from simple trading strategies via rebalancing at  discrete times and large order split.
In contrast, we contribute
a study of $M_1$-, $J_1$- and ucp-stability  for general approximations of c\`{a}dl\`{a}g  strategies  in a class of price impact models with  transient impact \eqref{eq:price process}, driven by quasi-left continuous martingales \eqref{eq:defbaseS}.

As a further contribution, and also to demonstrate the relevance and scope of the theoretical results, we discuss 
in the case of multiplicative impact 
a variety of 
examples 
where  continuity properties
play a 
 role.
In \cref{ex: optimal monotone liquidation in finite horizon} we establish existence of an optimal monotone liquidation strategy in finite time horizon using relative compactness and continuity of the proceeds functional in $M_1$.
\Cref{ex: optimal liquidation with general strategies} shows how to solve the optimal liquidation problem in infinite time horizon with non-negative bounded semimartingale strategies 
by approximating their proceeds via bounded variation strategies, here the $M_1$-stability being needed.
\Cref{ex: stochastic-finite-horizon} solves the liquidation problem  for an
 original extension of the model
where liquidity is stochastic and
 the time horizon is bounded by an expectation constraint for stopping times. 
This relies on $M_1$ convergence to define the trading proceeds. 
It provides an example of a liquidation problem where the optimum of singular controls is not attained in a class of finite variation strategies, but a suitable extension 
to semimartingale strategies is needed.
\Cref{ex: partial instantaneous impact} incorporates partially instantaneous recovery of price impact to our model. 
Herein, the $M_1$ topology plays the key role to identify (asymptotically realizable) proceeds as a continuous functional.
Last but not least,  \cref{sect:no arbitrage} proves absence of arbitrage for the large trader within a fairly large class of trading strategies.

The paper is organized as follows. \cref{sect:model} sets the model and defines the proceeds functional for finite variation strategies. 
In \cref{sect:continuity} we extend this definition to a more general set of strategies and prove our main \cref{thm:stability in j1 and m1}.
In the remaining \cref{sect:no arbitrage,sec:Examples} we concentrate on the case of multiplicative impact.
We show absence of arbitrage opportunities for the large investor in \cref{sect:no arbitrage} as a basis for a sensible financial model.
The examples related to optimal liquidation are investigated in \cref{sec:Examples}.

\section{A model for transient multiplicative price impact}\label{sect:model}

We consider a filtered probability space $(\Omega,\scF,(\scF_t)_{t\ge 0},\PP)$.
The filtration $(\scF_t)_{t\geq 0}$ is assumed to satisfy the usual conditions of right-continuity and completeness, with $\scF_0$ being the trivial $\sigma$-field.  Paths of semimartingales are taken to be  c\`{a}dl\`{a}g.
Let also $\scF_{0-}$ denote the trivial $\sigma$-field.
We consider a market with a single risky asset and a riskless asset (bank account) whose price is constant at $1$.
Without activity of large traders, the unaffected (discounted) price process of the risky asset would evolve according to the stochastic differential equation 
\begin{equation}\label{eq:defbaseS}
	\diff\baseS_t =  \baseS_{t-} (\xi_t \diff \langle M\rangle_t + \diff M_t)\,,\text{ }\qquad \baseS_0 > 0,
\end{equation}
where $M$ is a locally square-integrable martingale that is quasi-left continuous (i.e.\  for any finite predictable stopping time $\tau$, $\Delta M_\tau := M_{\tau} - M_{\tau-} = 0$  a.s.) with $\Delta M > -1$ and 
 $\xi$ is a predictable and bounded process. 
 In particular, 
the predictable quadratic variation process $\langle M \rangle$ is continuous \cite[Thm.~I.4.2]{JacodShiryaev2003_book}, and the unaffected (fundamental) price process  $\baseS>0$  can have jumps. We moreover assume that $\langle M \rangle
= \int_0^\cdot \alpha_s \diff s$ 
with density $\alpha$ being bounded (locally on compact time intervals) and whose paths  are (locally) Lipschitz, and that the martingale part of $\baseS$ is square integrable on compacts.
The assumptions on $M$ are satisfied e.g.\ for $M = \int \sigma\diff W$, where $W$ is a Brownian motion and $\sigma$ is a suitably regular bounded predictable process,
or for L\'{e}vy processes $M$ with suitable integrability and lower bound on jumps. 

To model the impact that trading strategies by a single large trader have on the risky asset price, let us denote by $(\assetsProcess_t)_{t\geq 0}$ her risky asset holdings throughout time and $\assetsProcess_{0-}$ be the number of shares the she holds initially.  
The process $\assetsProcess$  is the control strategy of the large investor who executes $\!\diff \assetsProcess_t$ market orders at time $t$ (buy orders if $\assetsProcess$ is increasing, sell orders if it is decreasing). 
We will assume throughout that strategies $\assetsProcess$ are predictable processes.
The large trader is faced with illiquidity costs because her trading has an adverse impact on the prices at which
her orders are executed as follows.
A \emph{market impact process} $Y$ (called volume effect process in \cite{PredoiuShaikhetShreve11}) captures the impact from a predictable strategy $\assetsProcess$ with càdlàg paths on the price of the risky asset, and is defined as the càdlàg adapted solution $Y$ to
\begin{equation} \label{eq:deterministic Y_t dynamics}
	\diff Y_t = -h(Y_t) \diff \langle M \rangle_t + \diff \assetsProcess_t
\end{equation}
for some initial condition $Y_{0-} \in \RR$.
We assume that $h: \RR \rightarrow \RR$ is Lipschitz with $h(0) = 0$ and $h(y)\sgn(y)\ge 0$ for all $y\in\RR$. The Lipschitz assumption on $h$ guarantees existence and uniqueness of $Y$ in a pathwise sense, see \cite[proof of Thm.~4.1]{PangTalrejaWhitt2007} and \cref{prop:cont of resilience} below. The sign assumption on $h$ gives \emph{transience} of the impact which recovers towards 0 (if $h(y) \neq 0$ for $y\neq 0$) when the large trader is inactive. 
The function $h$ gives the speed of resilience at any level of $Y_t$ and we will refer to it as \emph{the resilience function}.
For example, when $h(y) = \beta y$ for some constant $\beta > 0$, the market recovers at exponential rate (as in   \cite{ObizhaevaWang13,AlfonsiFruthSchied10,Lokka14}). Note that we also allow for $h\equiv 0$ in which case the impact is permanent as in \cite{BankBaum04}.
Clearly, the process $Y$ depends on $\assetsProcess$, and sometimes we will indicate this dependence as a superscript $Y = Y^\assetsProcess$. Some of the results in this paper could be extended with no additional work when considering additional noise in the market impact process, see the discussion in  \cref{ex: stochastic-finite-horizon}, or for less regular density $\alpha$ if the $-h(Y_t)\!\diff\langle M\rangle_t$ term in \eqref{eq:deterministic Y_t dynamics} is replaced by e.g.~$-h(Y_t)\!\diff t$.

If the large investor trades according to a continuous strategy $\assetsProcess$, the observed price $S$ at which infinitesimal quantities $\!\diff\assetsProcess$ are traded (see \eqref{eq:proceeds cont proc}) is given via \eqref{eq:deterministic Y_t dynamics} by
\begin{equation}\label{eq:price process}
	S_t:= g(\baseS_t, Y_t)\,,
\end{equation}
where the \emph{price impact function} $(x,y) \mapsto g(x,y)$ is $C^{2,1}$ and non-negative with $g_{xx}$ being locally Lipschitz in $y$, meaning that on every compact interval $I \subset \RR$ there exists $K > 0$ such that $\abs{g_{xx}(x,y) - g_{xx}(x,z)} \le K\abs{y-z}$ for all $x,y,z\in I$.
Moreover, we assume $g(x,y)$ to be non-decreasing in both $x$ and $y$.
In particular, selling (buying) by the large trader causes the price $S$ to decrease (increase).
This price impact is transient due to \eqref{eq:deterministic Y_t dynamics}.

\begin{example} \label{ex: additive or multiplicative impact}
	\cite{BankBaum04} consider a family of semimartingales $(S^\theta)_{\theta\in \RR}$ being parametrized by the large trader's risky asset position $\theta$. In our setup, this corresponds to general price impact function $g$ and $h\equiv 0$, meaning that impact is permanent.
	A known example in the literature on transient price impact is the additive case, $S = \baseS + f(Y)$, where 
	\cite{ObizhaevaWang13} take $f(y) = \lambda y$ to be linear,  motivated from a block-shaped limit order book. 
	For 
generalizations to non-linear increasing $f : \RR \to [0,\infty)$, see \cite{AlfonsiFruthSchied10,PredoiuShaikhetShreve11}.
	Note 
that  we require  $0\le g \in C^{2,1}$ for \cref{thm:stability in j1 and m1}, see \cref{rmk:positive prices}. 
	A  (somewhat technical) modification of the 
model by \cite{ObizhaevaWang13}, that fits with our setup and ensures positive asset prices,   
could be to  take   $g(\baseS, Y) = \varphi(\baseS + f(Y))$ with a non-negative increasing $\varphi \in C^2$ satisfying $\varphi(x) = x$ on $[\varepsilon, \infty)$ and $\varphi(\cdot) = 0$ on $(-\infty, -\varepsilon]$ for some $\varepsilon>0$.
	A different example, that   naturally ensures positive asset prices and will serve as our prime example for \cref{sect:no arbitrage,sec:Examples}, is multiplicative impact  $S = f(Y) \baseS$ for $f$ being strictly positive, non-decreasing, and with  $f\in C^1$ (to satisfy the conditions on $g$). Also here, the function $f$ can be interpreted as resulting from 
a limit order book, see \cite[Sect.~2.1]{BechererBilarevFrentrup2016-deterministic-liquidation}. 
\end{example}

While impact and resilience are given by general non-parametric functions,  note that these are static.
Considering such a model as a low (rather than high) frequency model, we do consider  approximations by continuous and finite variation strategies to be relevant.
To start,  let  $\assetsProcess$ be a continuous  process of finite variation (f.v., being adapted).
Then, the cumulative proceeds (negative expenses), denoted by $L(\assetsProcess)$, that are the variations in the bank account to finance buying and selling of the risky asset according to the strategy, can be defined (pathwise) in an unambiguous way.
Indeed, 
 proceeds  over period $[0,T]$ from a strategy $\Theta$ that is continuous should  be (justified also by \cref{lemma:proceeds of cont fv strategies})
\begin{equation}\label{eq:proceeds cont proc}
	L_T(\assetsProcess) := - \int_0^T S_u \diff \assetsProcess_u 
		= -\int_0^T g(\baseS_u, Y_u) \diff \assetsProcess_u.
\end{equation}
Our main task is to extend by stability arguments the model from continuous to more general trading strategies, in particular such involving block trades and 
even more general ones with c\`{a}dl\`{a}g paths, assuming transient price impact but no further frictions, like e.g.\ bid-ask spread (cf.~\cref{rmk:full LOB model}). To this end, we will adopt the following point of view: approximately similar trading behavior should yield similar proceeds. 
The next section will make precise what we mean by ``similar'' by considering different topologies on the c\`{a}dl\`{a}g path space.
It turns out that the natural extension of the functional $L$ from the space of continuous f.v.~paths to the space of c\`{a}dl\`{a}g f.v.~paths which makes the functional $L$ continuous in all of the considered topologies is as follows: 
for discontinuous trading we take the proceeds from a block market buy or sell order of size $\abs{\Delta \assetsProcess_\tau}$, executed immediately at a predictable stopping time $\tau<\infty$, to be given by
\begin{equation} \label{eq:block sale proceeds}
	-\int_0^{\Delta\assetsProcess_\tau} g(\baseS_{\tau-}, Y_{\tau-} + x) \diff x,
\end{equation}
and so the proceeds up to $T$ from a f.v.\  strategy $\assetsProcess$ with continuous part $\assetsProcess^c$  are
\begin{equation}\label{repeat eq:proceeds fv strategy}
	L_T(\assetsProcess):= -\int_0^T g(\baseS_u, Y_u) \diff \assetsProcess^c_u - \sum_{\substack{\Delta \assetsProcess_t \neq 0 \\ 0\leq t\leq T}} \int_0^{\Delta\assetsProcess_t} g(\baseS_{t-}, Y_{t-} + x) \diff x.
\end{equation}

Note that a block sell order means that $\Delta \assetsProcess_t < 0$, so 
the average price per share for this trade satisfies 
\(
	S_{t} \le -\frac{1}{\Delta \assetsProcess_t}\int_{0}^{\Delta \assetsProcess_t} g(\baseS_{t}, Y_{t-} + x) \diff x \le S_{t-}.
\)
Similarly, the average price per share for a block buy order, $\Delta \assetsProcess_t > 0$, is between $S_{t-}$ and $S_t$.
The expression in \eqref{eq:block sale proceeds} could be justified from a limit order book perspective for some cases of $g$, as noted in \cref{ex: additive or multiplicative impact}. But  we will derive it in the next section using stability considerations.

\begin{remark}
The aim to define a model for trading under price impact for general strategies is justified by applications in finance, which encompass trade execution, utility optimization and hedging.
While also e.g.\ \cite{BankBaum04,BluemmelRheinlaender17,CetinJarrowProtter04} define proceeds for semimartingale strategies, their definitions are  not ensuring continuity in the $M_1$ topology, in contrast to \cref{thm:stability in j1 and m1}. Another difference to \cite{BankBaum04,BluemmelRheinlaender17} is that our presentation is not going to rely on non-linear stochastic integration theory due to Kunita or, respectively, Carmona and Nualart.
\end{remark}

\section{Continuity of the proceeds in various topologies}\label{sect:continuity}

In this section we will discuss questions about continuity of the proceeds process $\assetsProcess \mapsto L_\cdot(\assetsProcess)$ with respect to various topologies: the ucp topology and the Skorokhod $J_1$ and (in particular) $M_1$ topologies. Each one captures different stability features, the suitability of which may vary with application context. 

Let us observe that for a continuous bounded variation trading strategy $\assetsProcess$ the proceeds from trading should be given by \eqref{eq:proceeds cont proc}. To this end, let us make just the assumption that
\begin{equation}\label{assumption on block trades}
\begin{aligned}
	&\text{a block order of a size }\Delta \text{ at some (predictable) time }t \text{ is executed at some  } \\
	&\text{average price per share  which is between } S_{t-} = g(\baseS_t, Y_{t-}) \text{ and }  g(\baseS_t, Y_{t-}+ c\Delta) 
\end{aligned}
\end{equation}
 for some  constant $c\ge 0$.
The assumption looks natural for $c=1$ where $Y_t= Y_{t-}+ c\Delta$, stating that a block trade is executed at an average price per share that is somewhere between the asset prices observed immediately before and after the execution.
The more general case $c\ge 0$ is just technical at this stage but will be needed in \Cref{ex: partial instantaneous impact}.
Assumption~\eqref{assumption on block trades} means that proceeds by a simple strategy as in \eqref{simple strategy n} are  
\begin{equation}
	L_t(\assetsProcess^n) = - \sum_{k:\ t_k \le t}\xi_k (\assetsProcess_{t_{k}} - \assetsProcess_{t_{k-1}})
\end{equation}
for some random variable $\xi_k$  between $g(\baseS_{t_k}, Y^{\assetsProcess^n}_{t_k -})$ and $g(\baseS_{t_k}, Y^{\assetsProcess^n}_{t_k - }+ c \Delta Y^{\assetsProcess^n}_t)$.
Note that at this point we have not specified the proceeds (negative expenses) from block trades, but we only assume that they satisfy some natural bounds.
Yet, this is indeed already sufficient to derive the functional \eqref{eq:proceeds cont proc} for continuous strategies as a limit of simple ones.
\begin{lemma}\label{lemma:proceeds of cont fv strategies}
	For $T >0$, approximate a continuous f.v.~process $(\assetsProcess_t)_{t\in [0,T]}$ by a  sequence $(\assetsProcess^n_t)_{t\in[0,T]}$ of simple trading strategies given as follows: For a sequence  of partitions 
$\{0 = t_0 < t_1 < \cdots < t_{m_n} = T\}$, $n\in \mathbb{N}$, with $\sup_{1\leq k \leq m_n}\abs{t_k - t_{k-1}} \to 0$ for $n \to \infty$, let
	\begin{equation}\label{simple strategy n}
		\assetsProcess^n_t := \assetsProcess_0 + \sum_{k = 1}^{m_n} \paren[\big]{ \assetsProcess_{t_k} - \assetsProcess_{t_{k-1}} } \indicator_{[t_k, T]}(t)\,,\quad t\in[0,T].
\end{equation}
Assume~\eqref{assumption on block trades} holds for some $c\ge 0$. Then 
	\(	\sup_{0\leq t\leq T} \abs{  L_t(\assetsProcess^n)+ \int_0^t S_u \diff \assetsProcess_u} \xrightarrow{n\rightarrow \infty} 0  \) a.s.
\end{lemma}

\begin{proof}
	Note that $\sup_{u\in [0,T]} \abs{\assetsProcess^n_u - \assetsProcess_u} \to 0$ as $n \to \infty$. The solution map $\assetsProcess\mapsto Y^\assetsProcess$ is continuous with respect to the uniform norm, see  \cref{prop:cont of resilience}.
	Therefore, 
	\begin{equation}\label{eq:tmp Y}
		\sup_{u\in [0,T]} \abs{ Y^{\assetsProcess^n}_u - Y^{\assetsProcess}_u } \to 0 \quad \text{a.s.\ for } n \to \infty. 
	\end{equation}
	Note that for $\Delta \assetsProcess_{t_k} := \assetsProcess_{t_{k}} - \assetsProcess_{t_{k-1}}$ and $\xi_k$ between $g(\baseS_{t_k}, Y^{\assetsProcess^n}_{t_k -})$ and $g(\baseS_{t_k}, Y^{\assetsProcess^n}_{t_k -} + c\Delta\assetsProcess_{t_k})$ and $Y:= Y^\assetsProcess$ we have
	\begin{align*}
		\abs{\xi_k - g(\baseS_{t_k}, Y_{t_k})} 
			&\leq L_g(\baseS_{t_k}, \omega) \max\braces[\big]{ \abs[\big]{ Y_{t_{k}}  - Y^{\assetsProcess^n}_{t_{k-}} - c\Delta \assetsProcess_{t_k} }, \abs[\big]{ Y_{t_k} - Y^{\assetsProcess^n}_{t_k-} } } 
	\\		&\leq \tilde c L_g(\baseS_{t_k}, \omega) \paren[\big]{ \abs[\big]{ Y_{t_k} - Y^{\assetsProcess^n}_{t_k} } + \abs{\Delta \assetsProcess_{t_k}} },
	\end{align*}
	where $\tilde c>0$ is a universal constant, $L_g(x, \omega)$ denotes the Lipschitz constant of $y \mapsto g(x,y)$ on a compact set, depending on the (bounded) realizations for $\omega \in \Omega$  of $Y^{\assetsProcess}$ and $Y^{\assetsProcess^n}$, $n\in \mathbb{N}$, on the interval $[0,T]$;
 such a compact set  exits since $\assetsProcess$ is continuous and $\sup_{u\in [0,T]} \abs[\big]{ Y^{\assetsProcess}_u- Y^{\assetsProcess^n}_u }$ can be bounded by a factor times the uniform distance between $\assetsProcess$ and $\assetsProcess^n$ on $[0,T]$, cf.\ \cite[proof of Thm.~4.1]{PangTalrejaWhitt2007}.
	Hence,
	\begin{align} \label{eq:L as RS integral plus error}
		L_t(\assetsProcess^n)  & =  - \sum_{k:\ t_k \le t} g(\baseS_{t_k}, Y^\assetsProcess_{t_k}) \paren[\big]{ \assetsProcess_{t_{k}} - \assetsProcess_{t_{k-1}}} + \scE^{n}_t\,,
	\\
	\hspace{-\mathindent}\text{where }
		\abs{\scE^n_t} &\leq \tilde c \paren[\Big]{ \sup_{u\in [0,T]} L_g(\baseS_u, \omega) } \sum_{k = 1}^{m_n} \paren[\big]{ \abs[\big]{ Y_{t_k} - Y^{\assetsProcess^n}_{t_k} } + \abs{ \Delta \assetsProcess_{t_k} } } \abs{ \Delta\assetsProcess_{t_k} }
	\\
		&\leq C(\omega) \paren[\Big]{ \sup_{1\leq k \leq m_n} \abs[\big]{ Y_{t_k} - Y^{\assetsProcess^n}_{t_k} } } \abs{ \assetsProcess(\omega)}_{\text{TV}} + C(\omega) \sum_{k=1}^{m_n} \abs{ \Delta \assetsProcess_{t_k} }^2 
	\\	
	&\to 0\quad \text{a.s.\ for }n \to \infty\ \text{(uniformly in $t$)},
	\end{align}
	thanks to \eqref{eq:tmp Y} and the fact that $\assetsProcess$ has continuous paths of finite variation.
	The claim follows since by dominated convergence the Riemann-sum process in \eqref{eq:L as RS integral plus error} converges a.s.\ to the Stieltjes-integral process $-\int_0^\cdot S_u \diff \assetsProcess_u$ uniformly on $[0,T]$.
\end{proof}

\begin{example}[Continuity issues for an alternative ``ad-hoc'' definition of proceeds] \label{ex:ad-hoc definition of proceeds}
	Consider the problem of optimally liquidating $\assetsProcess_{0-} = 1$ risky asset in time $[0,T]$ while maximizing expected proceeds.
	In view of assumption \eqref{assumption on block trades}, an alternative but possibly ``ad-hoc'' definition for proceeds $\tilde L_T$ of simple strategies could be to consider just some price for each block trade, similarly to \cite[Section~3]{BankBaum04} or \cite[Example~2.4]{HendersonHobson2011}.
	For multiplicative impact $g(\baseS,Y) = \baseS f(Y)$, taking e.g.\ the price directly after the impact would yield for simple strategies $\assetsProcess^n$ that trade at times $\{0 = t^n_0 < t^n_1 < \cdots < t^n_n = T\}$ the proceeds
	\(
		\tilde L_T(\assetsProcess^n) = -\sum_{k=0}^n \baseS_{t^n_k} f(Y^{\assetsProcess^n}_{t^n_k}) \Delta \assetsProcess^n_{t^n_k}
		\,.
	\)
	The family $(\assetsProcess^n)_n$ of strategies which liquidate an initial position of size $1$  until time $1/n$ in $n$ equidistant blocks of uniform size is given by $\assetsProcess^n_t := \sum_{k=1}^n \frac{n-k+1}{n} \indicator_{[ \frac{k-1}{n^2}, \frac{k}{n^2} )}(t)$.
	With unaffected price  $\baseS_t = e^{-\delta t} \widetilde{M}_t$ for a continuous martingale $\widetilde{M}$, and permanent impact ($h \equiv 0$), i.e.\ $Y_t = \assetsProcess_t - 1$, this yields $\EE[ \tilde L_T(\assetsProcess^n) ] \to \int_0^1 f(-y) \diff y$ for $n\to\infty$.
	Given $\delta \ge 0$,  for any non-increasing simple strategy $\assetsProcess = \sum_{k=1}^n \assetsProcess_{\tau_k} \indicator_{\rightOpenStochasticInterval{\tau_{k-1}, \tau_k}}$ with $\assetsProcess_{0-} = 1$ holds that
	$\EE[\tilde L(\assetsProcess)] \le \int_0^1 f(-y) \diff y$ with strict inequality for $\delta > 0$. 
	So the control sequence $(\assetsProcess^n)$ is only asymptotically optimal among all simple monotone liquidation strategies.
\end{example}

\begin{remark}\label{rmk:asymptontially realizable proceeds}
	Note that \cref{ex:ad-hoc definition of proceeds} is a toy example, since for permanent impact the optimal strategy (considering asymptotically realizable proceeds) is trivial and in case $\delta=0$ any strategy is optimal, cf.~\cite[Prop.~3.5(III) and  the comment preceding it]{GuoZervos13}.
	Nevertheless, this example shows that the object of interest are \emph{asymptotically realizable} proceeds, an
insight due to \cite{BankBaum04}.
	For analysis, it thus appears convenient and sensible not to make a formal distinction of (sub-optimal) realizable  and asymptotically realizable proceeds, but to consider the latter and interpret strategies accordingly.
		Investigating asymptotically realizable proceeds can  help to answer questions on modeling issues, e.g. whether the large investor could sidestep liquidity costs entirely and in effect act as a small investor, cf.\ \cite{BankBaum04, CetinJarrowProtter04}.
	One could impose, like \cite{CetinSonerTouzi10}, additional constraints on strategies to avoid such issues; 
       But in such tweaked models one could not investigate the effects from some given illiquidity friction alone, in isolation from other constraints, because results from an analysis will be consequences of the combination of both frictions. 
\end{remark}

By using integration-by-parts, we can obtain the following alternative representation of the functional in \eqref{eq:proceeds cont proc} for continuous f.v.~strategies:
\begin{align}
	L(\assetsProcess) 
		&= \int_0^\cdot G_x(\baseS_{u-}, Y^\assetsProcess_{u-}) \diff \baseS_u
		 + \int_0^\cdot \paren[\Big]{ \tfrac{1}{2}G_{xx}(\baseS_{u}, Y^\assetsProcess_u) \baseS_u^2 - g(\baseS_{u}, Y^\assetsProcess_u) h(Y^\assetsProcess_u) } \diff\angles{M}_u
	\notag\\
		&\quad - \paren[\big]{ G(\baseS_{\cdot}, Y^\assetsProcess_{\cdot}) - G(\baseS_0, Y^\assetsProcess_{0-}) }
	\notag\\
		&\quad + \sum_{\substack{\Delta\baseS_u \ne 0 \\ 0 \le u \le \cdot}} \paren[\big]{ G(\baseS_u, Y^{\assetsProcess}_u) - G(\baseS_{u-}, Y^{\assetsProcess}_u) - G_x(\baseS_{u-}, Y^{\assetsProcess}_u) \Delta\baseS_u }
		\,,
	\label{eq:def of proceeds process}
\end{align}
where $G(x,y):= \int_c^y g(x,z) \diff z$ for constant $c$, and using that $\baseS$ and $Y$ have no common jumps. The advantage of this representation is that the right-hand side of \eqref{eq:def of proceeds process} makes sense for any predictable process $\assetsProcess$ with c\`{a}dl\`{a}g paths in contrast to the term in \eqref{eq:proceeds cont proc}
This form of the proceeds will turn out to be helpful for the stability analysis. We will show that the right-hand side in \eqref{eq:def of proceeds process} is continuous in the control $\assetsProcess$ when the path-space of $\assetsProcess$, the c\`{a}dl\`{a}g path space, is endowed with various topologies. Hence, it can be used to define the proceeds for general trading strategies by continuity. Next section is going to discuss the topologies that will be of interest. 

\subsection{The Skorokhod space and its $M_1$ and $J_1$ topologies}

We are going to derive a continuity result (Theorem~\ref{thm:stability in j1 and m1}) for the functional $L$ in different topologies on the space $D\equiv D([0,T]) := D([0,T]; \RR)$ of real-valued c\`{a}dl\`{a}g paths on the time interval $[0,T]$.
Following the convention by \cite{Skorokhod56}, we take each element in $D[0,T]$ to be left-continuous at time $T$.\footnote{This is implicitly assumed also in \cite{Whitt2002_book}, see the compactness criterion in Thm.~12.12.2 which is borrowed from \cite{Skorokhod56}.} One could also consider initial and terminal jumps by extending the paths, see \cref{rmk: extended paths}.
At this point, let us remark that finite horizon $T$ is not essential for the results below, whose analysis carries over to the time interval $[0,\infty)$ because the topology on $D([0,\infty))$ is induced by the topologies of $D([0,T])$ for $T \geq 0$. More precisely, for the topologies we are interested in, $x_n\rightarrow x$ as $n\rightarrow \infty$ in $D([0,\infty))$ if $x_n\rightarrow x$  in $D([0,t])$ for the restrictions of $x_n, x$ on $[0,t]$, for any $t$ being a continuity point of $x$, see \cite[Sect.~12.9]{Whitt2002_book}.

Convergence in the uniform topology is rather strong, in that  approximating a path with a jump is only possible if the approximating sequence has jumps  of comparable size at the same time. 
If one is interested in stability with respect to slight shift of the execution in time, then a familiar choice that also makes $D$ separable, the Skorokhod $J_1$ topology, might be appropriate; for comprehensive study, see \cite[Ch.~3]{Billingsley99}. 
However, also here an approximating sequence for a path with jumps needs jumps of comparable size, if only at nearby times.
To capture the occurrence of the so-called \emph{unmatched jumps}, i.e.~jumps that appear in the limit of continuous processes, another topology on $D$ is more appropriate, the Skorokhod $M_1$ topology. 
Recall that $x_n \rightarrow x$ in $(D, d_{M_1})$ if $d_{M_1}(x_n, x) \rightarrow 0$ as $n\rightarrow \infty$, with
\begin{equation}
	d_{M_1}(x_n, x) := \inf \braces[\big]{ \norm{u-u_n} \vee \norm{r-r_n} \bigm| (u, r)\in \Pi(x), (u_n, r_n)\in \Pi(x_n)}\,,
\end{equation}
where $\norm{\cdot}$ denotes the uniform norm on $[0,1]$ and $\Pi(x)$ is the set of all \emph{parametric representations} $(u,r):[0,1]\to \Gamma(x)$ of the completed graph (with vertical connections at jumps)  $\Gamma(x)$ of $x\in D$, see \cite[Sect.~3.3]{Whitt2002_book}. 
In essence, two functions $x,y \in D$ are near to each other in $M_1$ if one could run continuously a particle on each graph $\Gamma(x)$ and $\Gamma(y)$ from the left endpoint toward the right endpoint such that the two particles are nearby in time and space.
In particular, it is easy to see that a simple jump path could be approximated in $M_1$ by a sequence of absolutely continuous paths, in contrast to the uniform and the $J_1$ topologies.  
More precisely, we have the following 
\begin{proposition}\label{prop:conv of Wong-Zakai in M1}
	Let $x\in D([0,T])$ 
	and consider the Wong-Zakai-type approximation sequence $(x_n) \subset D([0,T])$ defined by $x_n(t) := n\int_{t-1/n}^{t}x(s)\diff s$, $t\in [0,T]$.  
	Then 
	\[
		x_n \rightarrow x \quad \text{for }n \rightarrow \infty, \quad \text{in } (D([0,T]), M_1).
	\]
\end{proposition}
\begin{proof}
	To ease notation, we embed a path $x$ in $D([0,\infty))$ and consider the corresponding approximating sequence for the extended path on $[0,\infty)$. The claim follows by restricting to the domain $[0,T]$, as $0$ and $T$ are continuity points of $x$, cf.~\cite[Sect.~12.9]{Whitt2002_book}. The idea is to construct explicitly parametric representations of $\Gamma(x)$ and $\Gamma(x_n)$ that are close enough. For this purpose, we need to add ``fictitious'' time to be able to parametrize the segments that connect jump points of $x$. Indeed, let $(a_k)$ be a fixed convergent series of strictly positive numbers and let $t_1, t_2,\ldots$ be the jump times of $x$ ordered such that $\abs{\Delta x(t_1)}\geq \abs{\Delta x(t_2)} \geq \ldots$ and $t_k <  t_{k+1}$ if $\abs{\Delta x(t_k)}= \abs{\Delta x(t_{k+1})}$. 
	\newcommand{\jumpTime}[1]{\delta({#1})}
	Set $\jumpTime{t} := \sum_{k} a_k \indicator_{\{t_k\leq t\}}$, the total ``fictitious'' time added to parametrize the jumps of $x$ up to time $t$.
	
  Consider the time-changes $\gamma_n(t) := n\int_{t-1/n}^t(\jumpTime{u} +u)\diff u$ and $\gamma_0(t) := \jumpTime{t} + t$, $t\geq 0$, together with their continuous inverses $\gamma_n^{-1}(s) :=\inf\{u > 0 \mid \gamma_n(u) > s\}$ for $s\geq 0$, $n\geq 0$.
	It is easy to check that we have 
	\begin{equation}
		\gamma_n^{-1}(s) - 1/n < \gamma_0^{-1}(s) < \gamma_n^{-1}(s) < \infty \quad \text{for } s\geq 0,
	\end{equation}
	because $\gamma_n(t) < \gamma_0(t) < \gamma_n(t + 1/n)$, cf.~\cite[Lemma~6.1]{KurtzPardouxProtter95}. 
	Consider the sequence
	\(
		u_n(s):= x_n(\gamma_n^{-1}(s))
	\)
	for $s\geq 0$ and let
	\[ 
		u(s) := \begin{cases}
				x(\gamma_0^{-1}(s)) & \text{if $\eta_1(s) = \eta_2(s)$},
			\\	x(\gamma_0^{-1}(s))\cdot \frac{s - \eta_1(s)}{\eta_2(s) - \eta_1(s)} + x(\gamma_0^{-1}(s)-)\cdot \frac{ \eta_2(s) - s}{\eta_2(s) - \eta_1(s)} & \text{if $\eta_1(s) \neq \eta_2(s)$},
			\end{cases}
	\]
	where $[\eta_1(s), \eta_2(s)]$ is the ``fictitious'' time added for a jump at time $t = \gamma_0^{-1}(s)$, i.e.\ $\eta_1(s):= \sup\{\tilde s \mid \gamma_0^{-1}(\tilde s) < \gamma_0^{-1}(s)\}$ and $\eta_2(s):= \inf\{\tilde s \mid \gamma_0^{-1}(\tilde s) > \gamma_0^{-1}(s)\}$, as in \cite[p.~368]{KurtzPardouxProtter95}. 
	Then \cite[Lemma~6.2]{KurtzPardouxProtter95} gives $\lim_{n \to \infty}u_n = u$, uniformly on bounded intervals; our setup corresponds to $f \equiv 1$ there, so our $u_n,u$ correspond to  $V^{1/n}, V$ there.
	
	Now the claim follows by observing that $(u_n, \gamma^{-1}_n)$ is a parametric representation of the completed graph of $x_n$, i.e.~$(u_n, \gamma^{-1}_n) \in \Pi(x_n)$, and $(u, \gamma_0^{-1}) \in \Pi(x)$ which are arbitrarily close when $n$ is big.
\end{proof}
\begin{remark}\label{rmk:M1 important}
	A direct corollary of \cref{prop:conv of Wong-Zakai in M1} is that $D([0,T])$ is the closure of the set of absolutely continuous functions  in the Skorokhod $M_1$ topology, in contrast to the uniform or Skorokhod $J_1$ topologies where a jump in the limit can only be approximated by jumps of comparable sizes.  
\end{remark}
\begin{remark}[Extended paths] \label{rmk: extended paths}
	To include trading strategies that could additionally have initial and terminal jumps in our analysis, one may embed the paths of such strategies in the slightly larger space $D([-\varepsilon, T+\varepsilon];\RR)$ for some $\varepsilon> 0$, e.g.\ $\varepsilon = 1$, by setting $x(s)= x(0-)$ for $s\in [-\varepsilon, 0)$ and $x(s) = x(T+)$ for $s\in (T, T+\varepsilon]$; we will refer to thereby embedded paths as \emph{extended paths}.
	This extension is relevant when trying to approximate jumps at terminal time by absolutely continuous strategies in a non-anticipative way as e.g.\ in \cref{prop:conv of Wong-Zakai in M1} where it is clear that a bit more time could be required after a jump occurs in order to approximate it.  In particular, by considering extended paths the result of \cref{prop:conv of Wong-Zakai in M1} holds if one allows for initial and terminal jumps of $x$, but convergence holds in the extended paths space.
\end{remark}

\subsection{Main stability results}
\label{sec:main stability results}

Our main result is stability of  the functional $L$ defined by the right-hand side of \eqref{eq:def of proceeds process} for processes $\assetsProcess$ with c\`{a}dl\`{a}g paths.
\begin{theorem}\label{thm:stability in j1 and m1}
Let a sequence of predictable processes $(\assetsProcess^n)$ converge to the predictable process $\assetsProcess$ in $(D, \rho)$, in probability, where 
$\rho$ denotes the uniform topology, the Skorokhod $J_1$ or $M_1$ topology, being generated by a suitable metric $d$.
Assume that $(\assetsProcess^n)$ is 
bounded in $L^0(\PP)$, i.e.\ there exists $K\in L^0(\PP)$ such that $\sup_{0\leq t \leq T}|\assetsProcess^n_t|\leq K$ for all $n$.
	Then the sequence of processes $L(\assetsProcess^n)$ converges  to $L(\assetsProcess)$  in $(D, \rho)$ in probability, i.e.\
	\begin{equation}\label{eq:stability in metric}
		\PP\brackets[\big]{ d\bp{L(\assetsProcess^n), L(\assetsProcess)}\geq \varepsilon } \to 0 \qquad \text{  for }n \to \infty \text{ and   $\varepsilon > 0$}.
	\end{equation}
In particular, there is a subsequence $L(\assetsProcess^ {n_k})$ that converges a.s.\ to $L(\assetsProcess)$ in $(D, \rho)$.
\end{theorem}
Note that e.g.\ for almost sure convergence $\assetsProcess^n \to \assetsProcess$ in $(D,\rho)$, the $L^0(\PP)$ boundedness condition is automatically fulfilled.
\begin{proof}
By considering subsequences, one could assume that the sequence $(\assetsProcess^n)$ converges to $\assetsProcess$ in $(D, \rho)$ a.s.
The idea for the proof is to show that each summand in the definition of $L$ is continuous.
But as $D$ endowed with $J_1$ or $M_1$ is not a topological vector space, since addition is not continuous in general, further arguments will be required.
Addition is continuous (and hence also multiplication) if for instance the summands have no common jumps, see \cite[Prop.~VI.2.2]{JacodShiryaev2003_book} for $J_1$ and \cite[Cor.~12.7.1]{Whitt2002_book} for $M_1$.
In our case however, there are three terms in $L$ that can have common jumps, namely the stochastic integral process $\int_0^\cdot G_x(\baseS_{u-}, Y_{u-})\diff \baseS_u$, the sum $\Sigma:= \sum_{u\le \cdot} \bp{ G(\baseS_u,Y_u) - G(\baseS_{u-}, Y_u) - G_x(\baseS_{u-},Y_u) \Delta\baseS_u }$ of jumps and the term $-G(\baseS, Y)$.
At jump times of $\assetsProcess$ (i.e.\ of $Y$) which are predictable stopping times, 
$\baseS$ does not jump since it is quasi-left continuous. Hence the only common jump times can be jumps times of $\baseS$ which are totally inaccessible.
If $\Delta \baseS_\tau \neq 0$, we have then $\Delta (\int_0^\cdot G_x(\baseS_{u-}, Y_{u-})\diff \baseS_u)_\tau = G_x(\baseS_{\tau-}, Y_\tau) \Delta \baseS_\tau$ and also $\Delta(-G(\baseS,Y))_\tau = -\bp{G(\baseS_\tau, Y_\tau) - G(\baseS_{\tau-}, Y_\tau) }$, because $\Delta Y_\tau = 0$ a.s.
Since moreover $\Delta \Sigma_\tau = G(\baseS_\tau, Y_\tau) - G(\baseS_{\tau-}, Y_\tau) - G_x(\baseS_{\tau-}, Y_\tau) \Delta \baseS_\tau$, one has cancellation of jumps at jump times of $\baseS$.
However, these are times of continuity for $Y$ and this will be crucial below to deduce continuity of addition on the support of $\bp{ \int_0^\cdot G_x(\baseS_{u-}, Y_{u-})\diff \baseS_u, \Sigma, -G(\baseS, Y) }$ in $(D,\rho)\times (D, \rho) \times (D, \rho)$.

	First consider the case of uniformly bounded sequence $(\assetsProcess^n)$.
	Then the  processes 
	\[
		\diff Y^n_t = -h(Y^n_t)\diff \langle M \rangle_t + \diff \assetsProcess^n_t, \quad Y^n_{0-} = y\,,
	\]
	are uniformly bounded, so we can assume 
	w.l.o.g.\ that $h$, $gh$, $G$, $G_x$ and $G_{xx}$ are $\omega$-wise Lipschitz continuous and bounded (it is so on the range of all $Y^n$, $Y$, which is contained in a compact subset of $\RR$).
	By \cref{prop:cont of resilience} we have $Y^n \to Y$ in $(D, \rho)$, almost surely. 
	This implies $(\baseS, Y^n) \to (\baseS, Y)$ almost surely, by absence of common jumps of $\baseS$ and $Y$, cf.\ \cite[Prop.~VI.2.2b]{JacodShiryaev2003_book} for $J_1$ and\footnote{Using the strong $M_1$ topology in $D([0,\infty);\RR^2)$.} \cite[Thm.~12.6.1 and 12.7.1]{Whitt2002_book} for $M_1$.
	By the Lipschitz property of $G$ and (for the $M_1$ case) monotonicity of $G(\cdot,y)$ and $G(x,\cdot)$, we get
	\begin{equation} \label{eq: convergence of G(S,Yn)}
		G(\baseS, Y^n) \to G(\baseS, Y) \quad\text{in $(D, \rho)$, a.s.}
	\end{equation}
	Indeed, for the $M_1$ topology, it is easy to see that $(G(u^1,u^2),r) \in \Pi(G(\baseS,Y))$ for any parametric representation $((u^1,u^2),r)$ of $(\baseS,Y)$,
	because at jump times $t$ of $G(\baseS, Y)$, $z\mapsto r(z) \equiv t$ is constant on an interval $[z_1,z_2]$, and either $u^1$ or $u^2$ is constant on $[z_1,z_2]$.
	
	Note that jump times of $\assetsProcess$ and $Y$ coincide, and form a random countable subset of $[0,T]$. 
	Moreover, convergence in $(D, \rho)$ implies local uniform convergence at continuity points of the limit (for $\rho$ being the $M_1$ topology, cf.\ \cite[Lemma~12.5.1]{Whitt2002_book}, for the $J_1$ topology cf.\ \cite[Prop.~VI.2.1]{JacodShiryaev2003_book}).
	Hence, $Y^n_t \to Y_t$ for almost all $t\in [0,T]$, $\PP$-a.s.
	By Lipschitz continuity of $G_{xx}$ and $gh$, we get
	\(
		\tfrac{1}{2}G_{xx}(\baseS_t,Y^n_t) - g(\baseS_t, Y^n_t)h(Y^n_t) \to \tfrac{1}{2}G_{xx}(\baseS_t,Y_t) - g(\baseS_t, Y_t)h(Y_t)
	\), for almost-all $t \in [0,T]$, $\PP$-a.s.
	By dominated convergence, we conclude that
	\[
		\int_0^\cdot \!\! \paren[\big]{ \tfrac{1}{2}G_{xx}(\baseS_u,Y^n_u) - g(\baseS_u, Y^n_u)h(Y^n_u) } \diff \angles{M}_u 
		\to \!\int_0^\cdot \!\! \paren[\big]{ \tfrac{1}{2}G_{xx}(\baseS_u,Y_u) - g(\baseS_u, Y_u)h(Y_u) } \diff\angles{M}_u
	\]
	uniformly on $[0,T]$, a.s., using that $\langle M \rangle$ is 
absolutely continuous w.r.t.~Lebesgue measure. 
	Hence these two summands in the definition of $L$, see \eqref{eq:def of proceeds process}, are  ($\omega$-wise) continuous in $\assetsProcess$.
	
	Now we treat the stochastic integral and jump terms in \eqref{eq:def of proceeds process}. 
	By the above arguments we can also deal with the drift in the process $\baseS$. 
	Thus we may assume w.l.o.g.\ that $\baseS$ is a martingale.
	In particular, up to a localization argument (see below for details), we can assume that $\baseS$ is bounded and therefore the stochastic integral is a true martingale, since the integrand is bounded.
	Having $Y^n \to Y$ a.e.\ on the space $(\Omega\times [0,T], \PP\otimes\text{Leb}([0,T]))$, we can conclude convergence of the stochastic integrals in the uniform topology, in probability. 
	Dominated convergence on $\bp{ [0,T], \text{Leb}([0,T]) }$ yields
	\[
		\int_0^T (Y^n_{u-} - Y_{u-})^2 \diff \angles{\baseS}_u \to 0 \quad \text{as }n \to \infty, \ \PP\text{-a.s.}
	\] 
	Since $Y^n, Y$ are uniformly bounded one gets,
	again by dominated convergence, that
	\[
		\EE\brackets[\Big]{ \int_0^T (Y^n_{u-} - Y_{u-})^2 \diff \angles{\baseS}_u } \to 0 \quad \text{as }n \to \infty,
	\]
	i.e. $Y^n_{-} \rightarrow Y_{-}$ in $L^2(\Omega\times [0,T], \mathrm{d}\PP \otimes  \mathrm{d} \angles{\baseS})$. 
	By localization (to bound $\baseS$ and use that $G_x(x,y)$ is locally Lipschitz in $y$), It\^o's isometry and Doob's inequality, we get
	\begin{equation}\label{conv of stoch integr}
		\PP\brackets[\bigg]{ \sup_{0\leq t\leq T} \abs[\Big]{\int_0^t G_x(\baseS_{u-}, Y^n_{u-})\diff \baseS_u - \int_0^t G_x(\baseS_{u-}, Y_{u-})\diff \baseS_u } \geq \varepsilon } \to 0 \quad \text{as }n\to \infty.
	\end{equation}
	For the sum of jumps $\Sigma^n$ (defined like $\Sigma$, but with $Y^n$ instead of $Y$) we have a.s.\ uniform convergence $\Sigma^n \to \Sigma$ by \cref{lemma: uniform jump-sum convergence}.
	Hence $\int_0^t G_x(\baseS_{u-}, Y^n_{u-})\diff \baseS_u + \Sigma^n$ converges in ucp.
	To conclude on the proceeds, note that at jump times of $\baseS$, when cancellation of jumps occurs, one has continuity of $Y$ and hence local uniform convergence of the sequence $Y^n$. 
	For our setup, \cref{lemma: J1 cancellation,lemma: M1 cancellation} show continuity of addition on the support of $\bp{ \int_0^\cdot G_x(\baseS_{u-}, Y_{u-})\diff \baseS_u + \Sigma, -G(\baseS, Y) }$ (along the support of $\bp{ \int_0^\cdot G_x(\baseS_{u-}, Y^n_{u-})\diff \baseS_u + \Sigma^n, -G(\baseS, Y^n) }$) for the $J_1$ and $M_1$ topologies, respectively. 
	So the continuous mapping theorem \cite[Lem.~4.3]{Kallenberg02} yields the claim for the proceeds functional $L$ (the uniform topology being stronger than $\rho$).
	
It remains to investigate the more general case of $\baseS$ and $(\assetsProcess^n)$ being only bounded in $L^0(\PP)$. 
Note that the continuity of all terms except the stochastic integral in the definition of $L$ was proven $\omega$-wise; in this case $\sup_{n} \sup_{0\leq t\leq T} \abs{\assetsProcess^n_t(\omega)} < \infty$ (by the a.s.~convergence of $\assetsProcess^n$ to $\assetsProcess$ in $(D,\rho)$) and hence the same arguments carry over here by restricting our attention to compact sets (depending on $\omega$). 
Hence refinement of the argument above is only needed for the stochastic integral term. 
The bound on $\baseS$ and $(\assetsProcess^n)$ means that for every $\varepsilon > 0$ there exists $\Omega_{\varepsilon}\in \scF$ with $\PP(\Omega_\varepsilon) > 1- \varepsilon$ and a positive constant $K_\varepsilon$ which is a uniform bound for the sequence (together with the limit $\assetsProcess$) on $\Omega_\varepsilon$. 
For the stopping time $\tau := \inf \tau_n$, where $\tau_n := \inf\{t\geq 0\mid \abs{\assetsProcess^n_t} \vee \abs{\baseS_t} > K_\varepsilon \}\wedge T$ ($\tau$ is a stopping time because the filtration is right-continuous by our assumptions), we then have that $\tau = T$ on $\Omega_\varepsilon$. 
By the arguments above we conclude that $d\bp{ \int_0^{\cdot \wedge \tau} G_x(\baseS_{u-}, Y^n_{u-})\diff \baseS_u, \int_0^{\cdot \wedge \tau} G_x(\baseS_{u-}, Y_{u-})\diff \baseS_u } \to 0$ in probability. 
Since $\int_0^{\cdot\wedge\tau} G_x(\baseS_{u-}, Y^n_{u-})\diff \baseS_u = \int_0^\cdot G_x(\baseS_{u-}, Y^{n}_{u-})\diff \baseS_u$ on $\Omega_\varepsilon$, we conclude 
\[
	\PP\brackets[\bigg]{d\paren[\Big]{\int_0^\cdot G_x(\baseS_{u-}, Y^{n}_{u-})\diff \baseS_u, \int_0^\cdot G_x(\baseS_{u-}, Y_{u-})\diff \baseS_u}\geq \varepsilon} \leq 2\varepsilon
\]
for all $n$ large enough, and this finishes the proof since $\varepsilon$ was arbitrary.
\end{proof}

\begin{remark}
Inspection of the proof above reveals that predictability of the strategies is only needed to show why the addition map is continuous when there is cancellation of jumps in \eqref{eq:def of proceeds process}; indeed, for predictable $\assetsProcess$ the processes $Y^{\assetsProcess}$ and $\baseS$ will have no common jump and this was sufficient for the arguments. 
However, in the case when $M$ (and thus $\baseS$) is continuous, only one term in  \eqref{eq:def of proceeds process} might have jumps, namely $G(\baseS, Y^\assetsProcess)$. 
Hence, in this case the conclusion of \cref{thm:stability in j1 and m1} even holds under the relaxed assumption that the c\`{a}dl\`{a}g strategies are merely adapted, instead of being predictable.
\end{remark}

\begin{remark}\label{rmk:positive prices}
	Our assumption of positive prices (and monotonicity of $x \mapsto g(x,y)$) has been (just) used to prove the  $M_1$-convergence of $G(\baseS, Y^n)$ in \eqref{eq: convergence of G(S,Yn)}.
	If one would want to consider a model where prices could become negative (like additive impact $S = \baseS + f(Y)$, see  \cref{ex: additive or multiplicative impact}), then $M_1$-continuity of proceeds would not hold in general, as a simple counter-example can show.
	Yet, the above proof still shows $L_t(\assetsProcess^n) \to L_t(\assetsProcess)$ in probability, for all $t \in [0,T]$ where $\Delta\assetsProcess_t = 0$.
	Also note that for continuous $\assetsProcess^n$ converging in $M_1$ to a continuous strategy $\assetsProcess$, hence also uniformly, one obtains that proceeds $L(\assetsProcess^n) \to L(\assetsProcess)$ converge uniformly, in probability.
\end{remark}

An important consequence of \cref{thm:stability in j1 and m1} is a stability property for our model.
It essentially implies that we can approximate each strategy by a sequence of absolutely continuous strategies, corresponding to small intertemporal shifts of reassigned trades, whose proceeds will approximate the proceeds of the original strategy.
More precisely, if we restrict our attention to the class of monotone strategies, then we can restate this stability in terms of the Prokhorov metric on the pathwise proceeds (which are monotone and hence define measures on the time axis).
This result on stability of proceeds with respect to small intertemporal Wong-Zakai-type re-allocation of orders may be compared to seminal work by \cite{HindyHuangKreps92} on a different but related problem, who required that for economic reason the utility should be a continuous functional of cumulative consumption with respect to the  L\'evy--Prokhorov metric $d_{\text{LP}}$, in order to satisfy the sensible property of intertemporal substitution for consumption. Recall for convenience of the reader the definition of $d_{\text{LP}}$ in our context: for increasing c\`{a}dl\`{a}g paths on $[0,\tilde T]$, $x, y:[0,\tilde T]\rightarrow \RR$ with $x(0-) = y(0-)$ and $x(\tilde T) = y(\tilde T)$, 
\[
	d_{\text{LP}}(x,y) := \inf\{\varepsilon > 0\mid x(t)\leq y((t+\varepsilon)\wedge  \tilde T) + \varepsilon, \ \ y(t)\leq x((t+\varepsilon)\wedge \tilde T) + \varepsilon\ \ \forall t\in [0,\tilde T]\}.
\]

\begin{corollary}\label{cor:WZ proceeds converge}
	Let $\assetsProcess$ be a predictable
	process with c\`adl\`ag paths defined on the time interval $[0,T]$ (with possible initial and terminal jumps) that is extended to the time interval $[-1,T+1]$ as in \cref{rmk: extended paths}. 
	Consider 
	the sequence of f.v.~processes $(\assetsProcess^n)$ where 
	\begin{equation}\label{eq: WZ approx sequence}
		\assetsProcess^n_t := n\int_{t-1/n}^{t}\assetsProcess_s\diff s, \quad t\geq 0,
	\end{equation}
	and let $L := L(\assetsProcess) , L^n:= L(\assetsProcess^n)$ be the proceeds processes from the respective trading.
	Then $L^n_t \to L_t$ at all continuity points $t\in [0,T+1]$ of $L$ as $n\to \infty$, in probability.
	In particular, for any bounded monotone strategy $\assetsProcess$  the Borel measures $L^n(\mathrm d t;\omega)$ and $L(\mathrm d t;\omega)$ on $[0,T+1]$ are finite (a.s.) and converge in the L\'evy--Prokhorov metric $d_{\text{LP}}(L^n(\omega), L(\omega))$ in probability, i.e.~for any $\varepsilon > 0$,
	\[
		\PP\brackets[\big]{d_{\text{LP}}(L^n(\omega), L(\omega)) > \varepsilon } \to 0 \quad \text{as }n\rightarrow \infty.
	\]
\end{corollary}
\begin{proof}
	An application of \cref{prop:conv of Wong-Zakai in M1} together with \cref{thm:stability in j1 and m1} gives 
	\[
		d_{M1}(L^n, L) \xrightarrow{\PP} 0.
	\]
	The first part of the claim now follows from the fact that convergence in $M_1$ implies local uniform convergence at continuity points of the limit, see \cite[Lemma~12.5.1]{Whitt2002_book}. 
	The same property implies the claim about the L\'evy--Prokhorov metric because convergence in this metric is equivalent to weak convergence of the associated measures which on the other hand is equivalent to convergence at all continuity points of the cumulative distribution function (together with the total mass).
\end{proof}
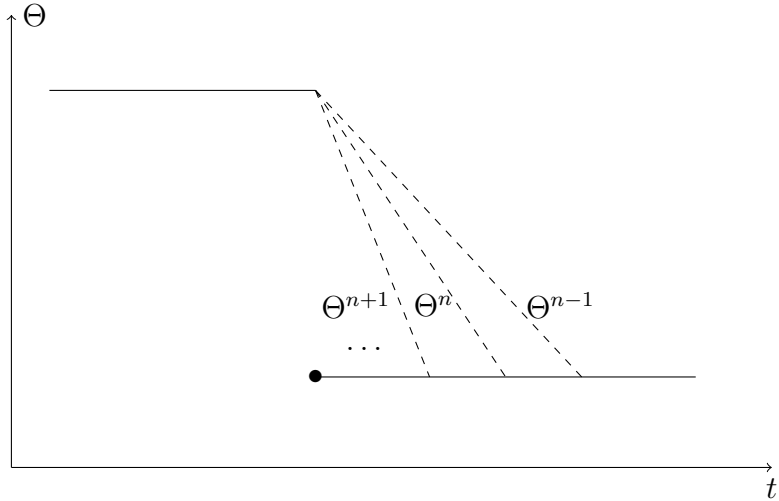
\begin{figure}[ht]%
	\centering
	\begin{tikzpicture}
		\coordinate (left)           at (0.5, 5);
		\coordinate (beforeJump)     at (4,   5);
		\coordinate (afterJump)      at (4,   1.2);
		\coordinate (right)          at (9,   1.2);
		\coordinate (afterDecreaseA) at (5.5, 1.2);
		\coordinate (afterDecreaseB) at (6.5, 1.2);
		\coordinate (afterDecreaseC) at (7.5, 1.2);
		
		\draw[->] (0,0) -- (10,0) coordinate[label = {below:$t$}] (xmax);
		\draw[->] (0,0) -- (0,6) coordinate[label = {right:$\assetsProcess$}] (ymax);
		
		\draw (left) -- (beforeJump);

		\node at (afterJump) {\textbullet};
		\draw (afterJump) -- (right);
		
		\draw[dashed] (beforeJump) -- (afterDecreaseA) node [near end, left] {$\assetsProcess^{{n+1}}$};
		\draw[dashed] (beforeJump) -- (afterDecreaseB) node [near end, left=-0.2em] {$\assetsProcess^{{n}}$};
		\draw[dashed] (beforeJump) -- (afterDecreaseC) node [near end, right] {$\assetsProcess^{{n-1}}$};
		
		\coordinate (midtimeJump) at ($0.5*(afterJump) + 0.5*(afterDecreaseA)$);
		\node at ($0.9*(midtimeJump) + 0.1*(beforeJump)$) {\dots};
		
	\end{tikzpicture}
	\caption{The Wong--Zakai approximation in \eqref{eq: WZ approx sequence} for a single jump process.}
	\label{fig:WZ}
\end{figure}

Note that the sequence $(\assetsProcess^n)$ from \cref{cor:WZ proceeds converge} satisfies $\assetsProcess^n \equiv \assetsProcess_T$ on $[T+1/n, T+1]$ for all $n$, i.e.~the approximating strategies arrive at the position $\assetsProcess_T$, however by requiring a bit more time to execute. Based on the Wong--Zakai approximation sequence from \eqref{eq: WZ approx sequence}, we next show that each semimartingale strategy on the time interval $[0,T]$ can be approximated by simple adapted strategies with uniformly small jumps that, however, again need slightly more time to be executed. 
\begin{proposition}
	Let $(\assetsProcess_t)_{t\in [0,T]}$ be a predictable process with c\`adl\`ag paths 
	extended to the time interval $[0,T+1]$ as in \cref{rmk: extended paths}. 
	Then there exists a sequence $(\assetsProcess^n_t)_{t\in [0,T+1]}$ of simple predictable c\`{a}dl\`{a}g processes with jumps of size not more than $1/n$ such that $d_{M1}(L(\assetsProcess^n), L(\assetsProcess)) \xrightarrow{\PP}0$ as $n\rightarrow \infty$, where $d_{M1}$ denotes the Skorokhod $M_1$ metric on $D([0,T+1];\RR)$. 
	Moreover, if $\assetsProcess$ is continuous, the same convergence holds true in the uniform metric on $[0,T]$ instead.
\end{proposition}
\begin{proof}
	Consider the Wong-Zakai-type approximation sequence $\widetilde \assetsProcess^n$ from \cref{cor:WZ proceeds converge} for which $d_{M_1}(L(\widetilde \assetsProcess^n), L(\assetsProcess))\xrightarrow{\PP}0$, where the Skorokhod $M_1$ topology is considered for the extended paths on time-horizon $[0,T+1]$. 
	Now we approximate each (absolutely) continuous process $\widetilde \assetsProcess^n$ by a sequence of simple processes as follows.

	For $\varepsilon > 0$, consider the sequence of stopping times with \(\sigma^{\varepsilon,n}_0 := 0\) and
	\begin{align*}
		\sigma^{\varepsilon,n}_{k+1} &:= \inf\braces[\big]{t \bigm| t > \sigma^{\varepsilon,n}_{k} \hbox{ and } \abs{\widetilde \assetsProcess^n_t - \widetilde\assetsProcess^n_{\sigma^{\varepsilon,n}_k}}  \ge \varepsilon } \wedge (\sigma^{\varepsilon,n}_k + 1/n) \quad \text{for $k\ge 0$}.
	\end{align*}
	Note that $\sigma^{\varepsilon,n}_k$ are predictable as hitting times of continuous processes and $\sigma^{\varepsilon,n}_k \nearrow \infty$ as $k\rightarrow \infty$ because the process $\widetilde\assetsProcess^n$ is continuous.
 When $\varepsilon \rightarrow 0$, we have $\assetsProcess^{\varepsilon,n}\xrightarrow{ucp}\widetilde\assetsProcess^n$  for 
	\[
		\assetsProcess^{\varepsilon,n} := \widetilde\assetsProcess^n_0 + \sum_{k = 1}^{\infty} \paren[\big]{ \widetilde\assetsProcess^n_{\sigma^{\varepsilon, n}_k} - \widetilde\assetsProcess^n_{\sigma^{\varepsilon, n}_{k-1}} } \indicator_{\rightOpenStochasticInterval{\sigma^{\varepsilon, n}_k, \infty}}.
	\] 
	Moreover, if for each integer $m\geq 1$ we define the (predictable) process $\assetsProcess^{\varepsilon,n,m}$ by 
	\[
		\assetsProcess^{\varepsilon,n,m} := \widetilde\assetsProcess^n_0 + \sum_{k = 1}^{m} \paren[\big]{ \widetilde\assetsProcess^n_{\sigma^{\varepsilon, n}_k} - \widetilde\assetsProcess^n_{\sigma^{\varepsilon, n}_{k-1}} } \indicator_{\rightOpenStochasticInterval{\sigma^{\varepsilon, n}_k, \infty}}
		\,,
	\] 
	then for each fixed $\varepsilon $ and $n$ we have $\assetsProcess^{\varepsilon,n,m} \xrightarrow{ucp} \assetsProcess^{\varepsilon,n}$ when $m\rightarrow \infty$.
	Hence, we can choose $\varepsilon = \varepsilon(n)$ small enough  and $m = m(n)$ big enough  such that
	\[
		d(\widetilde\assetsProcess^n, \assetsProcess^{\varepsilon(n), n, m(n)}) < 2^{-n},
	\]
	with $d(\cdot, \cdot)$ denoting a metric that metrizes ucp convergence (cf.\ e.g.~\cite[p.~57]{Protter04}).
	Thus, $\assetsProcess^n := \assetsProcess^{\varepsilon(n), n, m(n)}$ will be close to $\assetsProcess$ in the Skorokhod $M_1$ topology, in probability, because the uniform topology is stronger than the $M_1$ topology.

	Note that if $\assetsProcess$ is already continuous, no intermediate Wong-Zakai-type approximation would be needed, and so we obtain uniform convergence in probability in that case.
\end{proof}

The  previous theorem provided a general result on convergence in probability which relies solely on topological closeness of strategies.
Differently in spirit, an approximation idea due to \cite{BankBaum04} shows that one can actually approximate the proceeds of any strategy almost surely by some cleverly constructed continuous f.v.~strategies which can be implemented within the same time interval, if the base price $\baseS$ is continuous. 
\begin{proposition}[Almost sure uniform approximation~\emph{\`{a} la} Bank-Baum by continuous f.v.\ strategies]
	\label{prop:approx a la Bank-Baum}
	Suppose that $\baseS$ is continuous and $g(x,\cdot)$ and $h$ are continuously differentiable with locally H\"{o}lder-continuous derivatives for some index $\delta > 0$. 
	For any  predictable c\`{a}dl\`{a}g process $\assetsProcess$ on $[0,T]$ and any $\varepsilon>0$, there exists a continuous process $\assetsProcess^\varepsilon$ with f.v.~paths such $Y^\assetsProcess_T = Y^{\assetsProcess^\varepsilon}_T$, $\assetsProcess^\varepsilon_0 = \assetsProcess_{0-}$ and 
	$\abs{L_T(\assetsProcess) - L_T(\assetsProcess^\varepsilon)} \vee \abs{\assetsProcess_T - \assetsProcess^\varepsilon_T} \leq \varepsilon,\ \PP \text{-a.s.}$
\end{proposition}
\begin{proof}
Note that $K(y, t):= G(\baseS,y) - h(y)\int_0^t g(\baseS_u, y) \diff \langle M \rangle_u$ and $\widetilde K(y, t):= h(y)  \langle M \rangle_t$ define  smooth families of semimartingales in the sense of \cite[Def.~2.2]{BankBaum04} and 
\begin{equation}\label{eq:proceeds as non-linear integral}
	L_T(\assetsProcess) = \int_0^T K(Y_{s-}, \diff s) - \bp{ G(\baseS_T, Y_T) - G(\baseS_{0-}, Y_{0-}) }.
\end{equation}
Predictability of $\assetsProcess$ implies predictability of $Y$ and hence $Y_T$ is $\scF_{T-}$ measurable. By the multidimensional version of \cite[Thm.~4.4]{BankBaum04}  for the non-linear integrator $(K, \widetilde K)$ (extending the proof to this multidimensional setup is straightforward),  
for every $\varepsilon > 0$ there exists a predictable process $Y^\varepsilon$ with continuous paths of finite variation, such that $Y^\varepsilon_0 = Y_{0-}$, $Y^\varepsilon_T = Y_T$ and $\PP$-a.s. 
 \[
 	\sup_{0\leq t \leq T} \braces[\Big]{ \abs[\Big]{\int_{0}^t K(Y_{s-}, \diff s) - \int_{0}^t K(Y^\varepsilon_{s-}, \diff s)} \vee \abs[\Big]{ \int_0^t h(Y_s)\diff \langle M\rangle_s - \int_0^t h(Y^\varepsilon_s)\diff \langle M\rangle_s }}\leq \varepsilon.
 \]
 The process $Y^\varepsilon$ corresponds to a predictable process $\assetsProcess^\varepsilon$ with continuous f.v.~paths, namely $\assetsProcess^\varepsilon = Y^\varepsilon - Y_{0-} + \assetsProcess_{0-} + \int_0^\cdot h(Y^\varepsilon)\diff \langle M\rangle_u$, that satisfies $|\assetsProcess_T - \assetsProcess^\varepsilon_T| \leq \varepsilon$, and with reference to \eqref{eq:proceeds as non-linear integral}, also satisfies $|L_T(\assetsProcess) - L_T(\assetsProcess^\varepsilon)|\leq \varepsilon.$
\end{proof}

\subsection{Connection to the Marcus canonical equation}\label{sect:stability}

Here we explain briefly, how our proceeds functional connects with an interesting SDE  which is known as the Marcus canonical equation \cite{Marcus81}.
Stability in the sense of Wong--Zakai approximations for this kind of equations has been studied in \cite{KurtzPardouxProtter95}.
Their techniques offer an alternative way to derive the approximation result of \cref{cor:WZ proceeds converge}. 
Recently, stability of such equations for a $p$-variation rough paths variant of the $M_1$ topology has been studied in \cite{FrizChevyrev18}.

\begin{definition}[Marcus canonical equation]
	Let $\marcusDriver: \RR^d \to \RR^{d \times k}$ be continuously differentiable and $Z$ be a $k$-dimensional semimartingale.
	Then the notation
	\begin{equation} \label{eq:defMarcusIntegralNotation}
		X_t = X_{0-} + \int_0^t \marcusDriver(X_s) \circ\!\diff Z_s
	\end{equation}
	means that $X$ satisfies the stochastic integral equation
	\begin{align}
\nonumber
		X_t =& X_{0-} + \int_0^t \marcusDriver(X_{s-}) \diff Z_s
		 	+ \frac{1}{2} \sum_{j,m=1}^k \sum_{\ell = 1}^d \int_0^t \frac{\partial \marcusDriver_{\cdot,j}}{\partial x_\ell}(X_{s-}) \marcusDriver_{\ell,m}(X_{s-}) \diff {[Z^j, Z^m]^c_s}
 \label{eq:defMarcusIntegral}
		\\	&\quad + \sum_{\substack{0 \le s \le t \,,\, \Delta Z_s \ne 0}} \paren[\big]{\varphi(\marcusDriver(\cdot) \Delta Z_s, X_{s-}) - X_{s-} - \marcusDriver(X_{s-})\Delta Z_s},
	\end{align}
	where $\marcusDriver_{\cdot,j}$ is the $j^\text{th}$ column of $\marcusDriver$, $Z^j$ is the $j^\text{th}$ entry of $Z$ and $\varphi(\xi,x)$ denotes the value $y(1)$ of the solution to
	\begin{align} \label{eq:defPhiODE}
		y'(u) = \xi(y(u)) \quad \text{ with  }\quad y(0)=x.
	\end{align}
\end{definition}

\noindent 
The quadratic (co-)variation process is denoted by $[\cdot] = [\cdot]^c+[\cdot]^d$, it decomposes into a continuous part (appearing in (\ref{eq:defMarcusIntegral})) and a discontinuous part.
The next lemma gives a representation of the impact and proceeds processes of our model in terms of a Marcus canonical equation for the case $h \in C^1$.
To this end, let  the function $\marcusDriver: \RR^3 \rightarrow \RR^{3\times 3}$ for $X = (X^1, X^2, X^3)^{tr} \in \RR^3$ be given by
	\begin{align}\label{eq:def of Marcus driver}
		\marcusDriver(X) &:= \begin{pmatrix}
			- g(X^3, X^2) & 0 & 0
		\\	1 & 0 & -h(X^2)
		\\	0 & 1 & 0
		\end{pmatrix}.
	\end{align}

\begin{lemma} \label{lem:problem as Marcus integral}
	Let $\assetsProcess$ be a c\`{a}dl\`{a}g process with paths of finite total variation, and $L$ be defined by \eqref{repeat eq:proceeds fv strategy} be the process describing the evolution of proceeds generated by $\assetsProcess$.
	Set $X_t := \paren[\big]{ L_t, Y_t, \baseS_t }^{tr}$, so $X_{0-} = \paren[\big]{ 0, Y_{0-}, \baseS_{0-} }^{tr}$, and $Z_t := \paren[\big]{ \assetsProcess_t, \baseS_t, \angles{M}_t }^{tr} $.
	Then the process $X$ is the solution to the Marcus canonical equation 
	\[	
		X_t = X_{0-} + \int_0^t \marcusDriver(X_s) \circ\!\diff Z_s \,.
	\]
\end{lemma}
\noindent 
For the proof see \cref{sect:Marcus integral proofs}.
Following \cite[Sect.~6]{KurtzPardouxProtter95}, we now derive a Wong-Zakai-type approximation result in our setup.
For a bounded semimartingale process $\assetsProcess$ and $\varepsilon>0$ consider the  approximating  absolutely continuous processes defined by 
\begin{equation}\label{eq:theta eps}
	\assetsProcess^\varepsilon_t := \frac{1}{\varepsilon}\int_{t-\varepsilon}^{t}\assetsProcess_s\diff s, \quad t\geq 0,
\end{equation}
with the convention that $\assetsProcess_t = \assetsProcess_{0-}$ for $t < 0$. See \cref{fig:WZ}, where $\varepsilon = 1/n$.

Let $Z^\varepsilon_t := (\assetsProcess^\varepsilon_t, \baseS_t, \angles{M}_t)^{tr}$ and $X^\varepsilon$ be a solution to the following SDE in the Itô sense
\begin{equation}\label{eq:def of X^eps}
	\diff X^\varepsilon_t = \marcusDriver(X^\varepsilon_t) \diff Z^\varepsilon_t, \qquad X^\varepsilon_0 = X_{0-}.
\end{equation}
The next result on Wong-Zakai-type convergence is based on the theory from \cite[Sect.~5]{KurtzPardouxProtter95}.
See \cite[Thm.~6.2]{BechererBilarevFrentrup-arXiv-2015-old} for a proof in the case of $\baseS$ being geometric Brownian motion and $g(x,y) = x f(y)$, which however generalizes easily to continuous $\baseS$ and general impact function $g$.

\begin{theorem}\label{thm:Wong-Zakai approx}
	Suppose that $\baseS$ is continuous and 
	let $(\assetsProcess_t)_{t\geq 0}$ be a bounded semimartingale.
	For $\varepsilon > 0$, let $\assetsProcess^\varepsilon$ be the Wong-Zakai-type approximations from \eqref{eq:theta eps}. 
	Let $X^\varepsilon$ be defined by \eqref{eq:def of X^eps} for $Z^\varepsilon_t := (\assetsProcess^\varepsilon_t, \baseS_t, \angles{M}_t)^{tr}$ and $\marcusDriver$ as in \eqref{eq:def of Marcus driver}.
	For time-changes $\gamma_\varepsilon(t) := \frac{1}{\varepsilon}\int_{t-\varepsilon}^t([\assetsProcess]^d_s+s)\diff s$, consider the processes $(\scX^\varepsilon_t)_{t\geq 0}$ defined by $\scX^\varepsilon_t := X^\varepsilon_{\gamma_\varepsilon^{-1}(t)}$.
	For $\varepsilon \to 0$ the processes $\scX^\varepsilon$ then converge in probability in the compact uniform topology to a process $(\scX^0_t)_{t\geq 0}$, such that $X_t = (X_t^1,X_t^2,X_t^3)^{tr} := \scX^{0}_{\gamma_0(t)}$ is a solution of 
	\begin{equation}\label{eq:WZ limiting equation}
		X_t = X_{0-} + \int_0^{t}\marcusDriver(X_s) \circ\!\diff Z_s -  \paren[\Big]{\frac{1}{2} \int_0^t g_x(\baseS_{s}, X_{s-}^2) \diff [\baseS,\assetsProcess]_{s},\ 0,\ 0}^{tr},
	\end{equation}
	where $X_{0-} = (0, Y_{0-}, \baseS_0)^{tr}$ and $\gamma_0(t) := [\assetsProcess]^d_t + t$.
\end{theorem}

	\cref{thm:Wong-Zakai approx} directly gives, noting $X^1=L$, that for a bounded semimartingale strategy $\assetsProcess$, the proceeds $L=L(\assetsProcess)$ of this strategy up to  $T<\infty$ take the form
	\begin{align}  \label{eq:proceeds semimart strategies}
		L_T  = &- \int_0^Tg(\baseS_{t-}, Y^\assetsProcess_{t-}) \diff \assetsProcess_t  
			- \frac{1}{2} \int_0^T g_y(\baseS_{t-},  Y^\assetsProcess_{t-}) \diff\brackets{ \assetsProcess }^c_t 
			- \int_0^T g_x(\baseS_{t}, Y^\assetsProcess_{t-}) \diff [\baseS, \assetsProcess]_t \nonumber
	\\		&-\sum_{\substack{\Delta \assetsProcess_t \neq 0 \\ t\leq T}} \paren[\bigg]{\int_0^{\Delta \assetsProcess_t} g(\baseS_{t}, Y^\assetsProcess_{t-} + x)\diff x - g(\baseS_{t}, Y^\assetsProcess_{t-})\Delta \assetsProcess_t},
	\end{align}
	where the stochastic integral is understood in It\^{o}'s sense and $Y^\assetsProcess$ is given as in \eqref{eq:deterministic Y_t dynamics}. It is straightforward to see that \eqref{eq:proceeds semimart strategies} coincides with \eqref{eq:def of proceeds process}.

\begin{remark} \label{rmk:bounded semimartingale}
	a)
	Note that boundedness of $\assetsProcess$ implies that $X^2$ is bounded.
	Localizing along $\baseS$ (the variable $X^3$), we can assume that $g$ is globally Lipschitz continuous.
	This implies absolute convergence of the infinite sum in \eqref{eq:defMarcusIntegral}, see \cite[p.~356]{KurtzPardouxProtter95}.
	In particular, \eqref{eq:proceeds semimart strategies} is well-defined.

	b)
	The additional covariation term in the limiting equation \eqref{eq:WZ limiting equation} arises since only the strategies $\assetsProcess$ are approximated in a Wong--Zakai sense, but not also unaffected price $\baseS$ and clock $\angles{M}$.
	For strategies $\assetsProcess$ being of finite variation (as it would be natural under proportional transaction costs), this additional covariation term clearly vanishes.
	
	c) 
	Note that \cref{thm:Wong-Zakai approx} implies the results in \cref{cor:WZ proceeds converge} for bounded semimartingale processes $\assetsProcess$. 
	Indeed, \cref{thm:Wong-Zakai approx} gives for the first components $L^\varepsilon=X^{\varepsilon,1}$, $L=X^{0,1}$ that for any $\eta > 0$ and any horizon $T \in [0,\infty)$ we have
	\(
		\PP\brackets[\big]{ \sup_{t\le T} \abs[\big]{ L^\varepsilon_{\gamma_\varepsilon^{-1}(\gamma_0(t))} - L_t } \le \eta } \to 1
	\)
	for $\varepsilon\to 0$.
	Since $\gamma_\varepsilon^{-1}(\gamma_0(t)) \to t$ at continuity points of $\gamma_0$ (which are the continuity points of $\assetsProcess$ and thus of $L$) it follows that $\PP[\Omega^\eta_\varepsilon] \to 1$ as $\varepsilon \to 0$ with 
	\[
		\Omega^\eta_\varepsilon:= \{ \omega \mid \forall t\text{ with } \Delta L_t(\omega) =0 : \abs{L^\varepsilon_t(\omega) - L_t(\omega)} \le \eta \}.
	\]
	
	d) The proof of \cref{thm:Wong-Zakai approx} could be adapted to the case when $M$ is quasi-left continuous if the bounded semimartingale $\assetsProcess$ is assumed to be predictable.
\end{remark}

\section{Absence of arbitrage for the large trader}\label{sect:no arbitrage}

On the one hand the large trader is faced with adverse price reaction to her trades. 
On the other hand, her market influence might give her opportunities to manipulate price dynamics in her favor. 
It is therefore relevant to show that the model does not permit arbitrage opportunities for the large trader in a (fairly large) set of trading strategies. For this section we consider a multiplicative price impact model where $g(\baseS,Y) = f(Y)\baseS$ with a non-negative, increasing and continuously differentiable function $f$, cf.\ \cref{ex: additive or multiplicative impact}.\footnote{For additive dynamics of $\baseS$ instead of \eqref{eq:defbaseS}, one could carry out the analysis in this section also in the case of additive impact $g(\baseS,Y) = \baseS + f(Y)$}
Consider a 
portfolio $(\beta_t, \assetsProcess_t)$ of the large investor, where $\beta_t$ represents holdings in the bank account (riskless num\'{e}raire with discounted value $1$) and $\assetsProcess_t$ denotes holdings in the risky asset $S$ at time $t$.
We will consider bounded c\`{a}dl\`{a}g strategies $\assetsProcess$ on the full time horizon $[0,\infty)$ although our results below will deal with a finite but arbitrary horizon.
For the strategy $(\beta, \assetsProcess)$ to be self-financing, the bank account 
evolves according to 
\begin{equation}\label{eq:bank account}
	\beta_t = \beta_{0-} + L_t(\assetsProcess) \,, \quad  t\geq 0,
\end{equation}
with $L(\assetsProcess)$ as in \eqref{eq:def of proceeds process}.
In order to define the wealth dynamics induced by the large trader's strategy, we have to specify the dynamics of the value of the risky asset position in the portfolio.
If the large trader were to unwind her risky asset position at time $t$ immediately by selling $\Theta_t$ shares (meaning to buy shares in case of a short position  $\Theta_t<0$), 
 the resulting change in the bank account would be given by a term of the form \eqref{eq:block sale proceeds}.
In this sense, let the \emph{instantaneous liquidation value} process of her position be
\begin{equation}\label{eq:wealth process}
	V^\assetsProcess_t = \beta_t + \baseS_t\int_{0}^{\assetsProcess_t} f(Y^\assetsProcess_t - x)\diff x\,, \quad t\geq 0.
\end{equation}
This corresponds to the asymptotically realizable real wealth process in \cite{BankBaum04}.
Its dynamics~\eqref{eq:dynamics of V}  are mathematically tractable and relevant, e.g.\ to study no-arbitrage. For $F(x):= \int_0^x f(y)\diff y$ we have
	\(
		\baseS_t\int_{0}^{\assetsProcess_t} f(Y^{\assetsProcess}_{t}-x)\diff x = \baseS_t \paren[\big]{F(Y^\assetsProcess_t) - F(Y^\assetsProcess_t - \assetsProcess_t)}.
	\)
	By \eqref{eq:def of proceeds process} and \eqref{eq:bank account}, noting that
$Y^\assetsProcess - \assetsProcess$ 
and  $ \langle M\rangle$ are absolutely continuous processes,  we have
	\begin{align} \notag
		\hspace{-2em}
		\diff V^\assetsProcess_t 
			&= F(Y^\assetsProcess_{t-}) \diff \baseS_t - \baseS_t (f h)(Y^\assetsProcess_{t-}) \diff \langle M\rangle_t - \diff \paren[\big]{ \baseS_\cdot F(Y^\assetsProcess_\cdot - \assetsProcess_\cdot) }_t
		\\ \nonumber
			&= \paren[\big]{ F(Y^\assetsProcess_{t-}) - F(Y^\assetsProcess_{t-} - \assetsProcess_{t-}) }\diff \baseS_t - \baseS_t \paren[\big]{ F'(Y^\assetsProcess_{t-}) - F'(Y^\assetsProcess_{t-} - \assetsProcess_{t-})} h(Y^\assetsProcess_{t-}) \diff \langle M\rangle_t 
			\\  \label{eq:dynamics of V}
			&=  \paren[\big]{F(Y^\assetsProcess_{t-}) - F(Y^\assetsProcess_{t-} - \assetsProcess_{t-})} \baseS_{t-} (\mu_t \diff \langle M \rangle_t + \diff M_t),
	\end{align}
	\[
	\hspace{-\mathindent}\text{with }
		\mu_t := \xi_t -h(Y^\assetsProcess_{t-}) \cdot \frac{F'(Y^\assetsProcess_{t-}) - F'(Y^\assetsProcess_{t-} - \assetsProcess_{t-})}{F(Y^\assetsProcess_{t-}) - F(Y^\assetsProcess_{t-} - \assetsProcess_{t-})} \indicator_{\{\assetsProcess_{t-} \neq 0\}}
	\text{ and }
	V^\assetsProcess_0 \!=\! \beta_0 + \int_0^{\assetsProcess_0} \!\!\! f(Y_0 + x) \diff x.
	\]

We will prove a no-arbitrage theorem for the large trader essentially for models that do not permit arbitrage opportunities for small investors in the absence of 
trading by the large trader. More precisely, for this section we assume for the driving noise $M$ the
\begin{assumption}\label{assumption:NA}
For every predictable and bounded process $\mu$ and every $T\geq 0$, there exists a probability measure $\PP^\mu \approx \PP$ on $\scF_T$ such that the process $M + \int_0^\cdot \mu_s\diff \langle M\rangle_s$ is a $\PP^\mu$-local martingale on $[0,T]$.
\end{assumption} 

\begin{example}[Models satisfying assumption \cref{assumption:NA}]
a)  If $M$ is continuous, then under our model assumptions from \cref{sect:model}, for every predictable and bounded process $\mu$ the probability measure $\diff \PP^\mu = \mathcal{E}(-\int_0^\cdot \mu_s \diff M_s )\diff \PP$ is  well-defined (thanks to Novikov's condition) and  satisfies \cref{assumption:NA}.
	
b) Let $M$ be a L\'{e}vy process that is a martingale with  $\Delta M > -1$ and $\EE[M_1^2]<\infty$. 
In this case, it is a special semimartingale with characteristic triplet $(0, \sigma, K)$ (w.r.t.~the identity truncation function), and  we have the decomposition $M = \sqrt{\sigma} W + x*(\mu^M - \nu^\PP)$, where $W$ is a $\PP$-Brownian motion (or null if $\sigma=0$), $\mu^M$ is the jump measure of $M$ and $\nu^\PP(\mathrm d x, \mathrm d t) = K(\mathrm d x)\diff t$ is the $\PP$-predictable compensator of $\mu^M$.
We have $\langle M \rangle_t = \lambda t$, $t\geq 0$, for some $\lambda \geq 0$.
In the case $\sigma > 0$, \cref{assumption:NA} is clearly satisfied.
Indeed, an equivalent change of measure by the standard Girsanov's theorem with respect to the non-vanishing (scaled) Brownian motion $M^c$ can be done such that $M^c+\int \mu \diff \langle M\rangle$ becomes a martingale, 
without changing the L\'{e}vy measure. 

Otherwise, in  case of $\sigma = 0$, $M$ is a pure jump L\'{e}vy process.
For this case, let us restrict our consideration to the situation of two-sided jumps, since pure-jump L\'{e}vy processes  of such type appear more relevant to the modeling of financial returns than those ones with one-sided jumps only; examples are the exponential transform of the variance-gamma process or the so-called CGMY-process (suitably compensated to give a martingale exponential transform), cf.~\cite{KallsenShiryaev02,CGMY02} for the relevant notions and models respectively.
Here, it turns out  that $K((-\infty, 0)) >0$ and $K((0, +\infty)) > 0$  is already a sufficient condition for \cref{assumption:NA} to hold, i.e.\ possibility for jumps occurring in both directions.
Indeed, a suitable  change of measure can then be constructed as follows.
Let $n> 0$ be such that $K([1/n, n]) > 0$ and $K([-n, -1/n])>0$.
Denote $C^+ := \int_{[1/n, n]}x^2 K(\mathrm d x) > 0$ and $C^- := \int_{[-n, -1/n]}x^2 K(\mathrm d x) > 0$. Define functions $Y^\pm:\RR\to \RR$ by $Y^+ := 1$ on $[1/n, n]^c$, $Y^+(x) - 1 := x/C^+$ on $[1/n, n]$, and by  $Y^- := 1$ on $[-n, -1/n]^c$, $Y^-(x) - 1 := -x/C^-$ on $[-n, -1/n]$, respectively.
Thus $\int_\RR x(Y^\pm(x) - 1)K(\mathrm d x)= \pm 1$ and hence , with $\eta := \lambda \mu$, the bounded previsible process 
\[
	Y(\omega, t, x) := \eta^-_t(\omega) (Y^+(x) - 1) + \eta^+_t(\omega) (Y^-(x) - 1) + 1
\] 
satisfies $\int_\RR x(Y(x) - 1)K(\mathrm d x) = -\eta$.
The stochastic exponential $Z:=\mathcal{E}((Y-1)*(\mu^L - \nu^\PP))$ is a strictly positive $\PP$-martingale, cf.\ \cite[Prop.~5]{EscheSchweizer05}.
So for $T\geq 0$ there is a measure $d\PP^\mu= Z_T d\PP$ with density process $(Z_t)_{t\le T}$.
By Girsanov's theorem \cite[Thm.~III.3.11]{JacodShiryaev2003_book}, $M- 1/Z_-\cdot  \langle M, Z \rangle = M + \int_0^\cdot \mu_u \diff \langle M \rangle_u$ is a $\PP^\mu$-local martingale on $[0,T]$.
\end{example}

The set of \emph{admissible trading strategies} that we consider is
\begin{align*}
	\hspace{-0.5\mathindent}\mathcal{A} := \big\{(\assetsProcess_t)_{t\geq 0} \mid {}&\text{bounded, predictable, c\`{a}dl\`{a}g, with  $V^\assetsProcess$   bounded from below,}
\\  		&\text{ $\assetsProcess_{0-} = 0$, and such that $\assetsProcess_t = 0$ for $t\in[ T,\infty)$ for some $T<\infty$} 
 \big\}.
\end{align*}
Note that for such a strategy $\assetsProcess$ it clearly holds $V^\assetsProcess=\beta$ on $[T,\infty)$, i.e.~beyond some bounded horizon $T<\infty$ the liquidation value coincides with the cash holdings $\beta_T$.
Boundedness from below for $V^\assetsProcess$ has a clear economical meaning, while the boundedness of $\assetsProcess$ may be viewed as a more  technical requirement.
It ensures under \cref{assumption:NA}
the existence of a strategy-dependent measure $\QQ^\assetsProcess\approx \PP$ (on $\scF_T$) so that $V^\assetsProcess$ is a $\QQ^\assetsProcess$-local martingale on $[0,T]$. 
This relies on \eqref{eq:dynamics of V} and is at the key idea for the proof for 
\begin{theorem}\label{thm:absence of arbitrage}
	Under \cref{assumption:NA}, the model is free of arbitrage up to any finite time horizon $T \in [0,\infty)$, in the sense that there exists no 
	$\assetsProcess \in \mathcal{A}$ with $\assetsProcess_t=0$ on $t\in[T,\infty)$ such that for the corresponding self-financing strategy $(\beta, \assetsProcess)$ with $\beta_{0-} = 0$ we have 
\begin{equation}
\PP\brackets{ V^\assetsProcess_T \geq 0 } = 1
\quad\text{ and }\quad
\PP\brackets{ V^\assetsProcess_T > 0 } > 0\,. \label{eq: arbitrage opp}
\end{equation}
\end{theorem}
\begin{proof}
Recall the SDE \eqref{eq:dynamics of V} which describes the liquidation value process $V$, and note  that $V_0=0$.
	For each $\assetsProcess \in \mathcal{A}$ we have that $(\assetsProcess, Y^\assetsProcess)$ is bounded.
	Thus, the drift $\mu$ is bounded as well because, in the case of $\assetsProcess_{t-} \ne 0$, by the mean value theorem we have  
	\[
		\frac{ F'(Y^\assetsProcess_{t-}) - F'(Y^\assetsProcess_{t-} - \assetsProcess_{t-})}{F(Y^\assetsProcess_{t-}) - F(Y^\assetsProcess_{t-} - \assetsProcess_{t-})} = \frac{f'(z_1)}{f(z_2)} \quad \hbox{for some $z_{1,2}$ between $Y^\assetsProcess_{t-}$ and $Y^\assetsProcess_{t-} - \assetsProcess_{t-}$},
	\]
	and this is bounded from above because $f,f'$ are continuous and $f>0$ (so it is bounded away from zero on any compact set).
	Hence, \cref{assumption:NA} guarantees the existence of $\PP^\mu \approx \PP$ on $\scF_T$ such that $V^\assetsProcess$ is a $\PP^\mu$-local martingale on $[0,T]$, and since it is also bounded from below, it is a $\PP^\mu$-supermartingale, so $E^\mu[V^\Theta_T]\le V^\Theta_0=0$.  
	This rules out arbitrage opportunities,  as described in \eqref{eq: arbitrage opp}, under any probability $\PP$ equivalent to $\PP^\mu$  on ${\scF}_T$, for any $T\in [0,\infty)$.
\end{proof}

\begin{remark}[Extension to bid-ask spread] \label{rmk:full LOB model}
	Absence of arbitrage 
in the 
model
 with zero bid-ask spread naturally implies no arbitrage for model extensions with spread, at least when the admissible  trading strategies have paths of finite variation.
	To make this precise, let us model different impact processes  $Y^{\assetsProcess^-}$ and $Y^{\assetsProcess^+}$ from selling and buying, respectively, according to \eqref{eq:deterministic Y_t dynamics}, and best bid and ask price processes $(S^b\!, S^a) := \paren[\big]{ f(Y^{\assetsProcess^-})\baseS^b\!, f(Y^{\assetsProcess^+})\baseS^a }$ with $S^b \leq S^a$ for non-increasing $\assetsProcess^-$ and non-decreasing $\assetsProcess^+$.
	Then, the proceeds from implementing $(\assetsProcess^-, \assetsProcess^+)$ on $[0,T]$ would be
	\[
		-\int_0^T \! S^b_t \diff \assetsProcess^{-,c}_t 
		-\int_0^T \! S^a_t \diff \assetsProcess^{+,c}_t 
		-\sum_{\mathclap{\substack{0\leq t \leq T \\ \Delta \assetsProcess^-_t < 0 }}} \baseS^b_t \int_0^{\Delta \assetsProcess^-_t}\! f(Y^{\assetsProcess^-}_{t-} + x) \diff x 
		-\sum_{\mathclap{\substack{0\leq t \leq T \\ \Delta \assetsProcess^+_t >0 }}} \baseS^a_t \int_0^{\Delta \assetsProcess^+_t}\! f(Y^{\assetsProcess^+}_{t-} + x) \diff x.
	\] 
	Now for $\assetsProcess := \assetsProcess^- + \assetsProcess^+$, the initial relation $Y^{\assetsProcess^-}_{0-} \leq Y^\assetsProcess_{0-} \leq Y^{\assetsProcess^+}_{0-}$ implies $Y^{\assetsProcess^-} \leq Y^\assetsProcess \leq Y^{\assetsProcess^+}$.
	Hence $S^b \le S \le S^a$ and the proceeds above for the model with non-vanishing spread would be dominated (a.s.)\ by those that we get in \eqref{repeat eq:proceeds fv strategy}, i.e.~in the model without bid-ask spread.
	In an alternative but different variant, one could extend the zero bid-ask spread model to a one-tick-spread model, motivated by insights in \cite{ContLarrard13}, by letting $(S^b,S^a) := (S,S+\delta)$ for some $\delta>0$.
	Again, proceeds in this model would be dominated by those in the zero-spread model.
	In either variant, absence of arbitrage opportunities in the zero bid-ask spread model implies the same for an extended model with spread.
\end{remark}

\begin{remark}[Extension to c\`{a}gl\`{a}d strategies]\label{rmk:caglad strategies}
For any c\`{a}gl\`{a}d (left continuous with right limits) $(\assetsProcess_t)_{t\ge0}$ (with $\assetsProcess_{0-}=\assetsProcess_{0}$) the unique c\`{a}gl\`{a}d solution $Y^{\Theta}$ to
the integral equation $Y_t-Y_s= \int_s^t h(Y_u)\alpha_u\, du + \Theta_t-\Theta_s$ ($0\le s<t$, with $Y_0=Y_{0-}$), 
  corresponding to \eqref{eq:deterministic Y_t dynamics}, 
can be defined pathwise (cf.\ proof of \cite[Thm.~4.1]{PangTalrejaWhitt2007}); statements on
 c\`{a}dl\`{a}g paths ($\bar \assetsProcess$,$Y^{\bar \assetsProcess}$) translate to  c\`{a}gl\`{a}d paths ($\assetsProcess$,$Y^{\assetsProcess}$)
 by  relations $\bar \assetsProcess_{t-} = \assetsProcess_{t}$ and $Y^{\bar \assetsProcess}_{t-}=Y^{\assetsProcess}_t$, $t\ge 0$.
Using this, we can define the dynamics of the liquidation wealth process $V$
for any 
 strategy $\assetsProcess$ which is adapted with  c\`{a}gl\`{a}d paths or predictable with c\`{a}dl\`{a}g paths, and hence locally bounded, 
by the the unique (strong)
solution to the SDE \eqref{eq:dynamics of V} for given initial condition $V_0\in \RR$. 
Thereby, the result on absence of arbitrage can be extended to a larger set of strategies, which contains the set $\mathcal{A}$ and in addition all bounded adapted and c\`{a}gl\`{a}d (left-continuous with right limits) processes $(\Theta_t)_{t\ge 0}$ with 
$\assetsProcess_{0-}=\assetsProcess_{0} = 0$ for which  there exists some $T<\infty$ such that $\assetsProcess_t = 0$ for $t\in[ T,\infty)$ holds.
Indeed, the same lines of proof show that such $\Theta$ cannot give an arbitrage opportunity in the sense of \cref{thm:absence of arbitrage}.
\end{remark}

\section{Application examples}
\label{sec:Examples}

In this section, we present four examples in the framework of multiplicative impact $g(\baseS,Y) = f(Y)\baseS$, cf.\ \cref{ex: additive or multiplicative impact}, that highlight different questions in which our stability results are helpful. \Cref{ex: optimal monotone liquidation in finite horizon} shows, by compactness argument, the existence of an optimal control by an application of our continuity result in \cref{thm:stability in j1 and m1}. For this, it is rather easy to check that the set of controls is compact for the $M_1$ topology. In \cref{ex: optimal liquidation with general strategies} we identify the solution of an optimal liquidation problem with the already known optimizer in a smaller class of admissible controls, by approximating semimartingale strategies with strategies of bounded variation, where stability of the proceeds functional plays a crucial role.

\Cref{ex: stochastic-finite-horizon,ex: partial instantaneous impact} illustrate modifications of the price impact model by changing the impact process to allow stochastic, respectively partially instantaneous, impact, to which the analysis in \cref{sect:continuity} carries over.
Herein, the $M_1$ topology is again key for identifying the (asymptotically realizable) proceeds and thus extending the models to a larger class of trading strategies. This is particularly crucial in \Cref{ex: stochastic-finite-horizon}, where the optimal liquidation problem  with stochastic liquidity can be solved explicitly 
by a convexity argument if the price process is  a martingale. In this case, any finite-variation strategy turns out to be suboptimal.
We construct an optimal singular control of infinite 
variation.

\subsection{Optimal liquidation problem on finite time horizon} \label{ex: optimal monotone liquidation in finite horizon}
In this example,  using continuity of the proceeds in the $M_1$ topology we will show that the optimal liquidation problem over monotone strategies on a finite time horizon admits an optimal strategy. For $\theta \geq 0$ shares to be liquidated, the problem is to
\begin{equation}\label{eq:opt liq problem monotone}
	\text{maximize}\quad \EE[ L_T(\assetsProcess) ] \quad \text{over } \assetsProcess\in \admissibleSellStrategies{\theta},
\end{equation}
over the set of all decreasing adapted càdlàg $\assetsProcess$ with $\assetsProcess_{0-} = \theta$ and $\assetsProcess\indicator_{[T,\infty)} = 0$.
We consider the situation when the unaffected price process has constant drift, i.e.~$\baseS_t = e^{\mu t} M_t$ for $t\geq 0$, where $\mu \in \RR$ and $M$ is a non-negative continuous martingale that is locally square integrable. 
Existence and (explicit) structural description of the optimal strategy is already known in the following two cases:
a) $\mu = 0$ and any time horizon $T\geq 0$, cf.~\cite{PredoiuShaikhetShreve11,Lokka12};
or: b)  $\mu < 0$ and sufficiently big time horizon $T \geq T(\theta, \mu)$ under additional assumptions on $f$ and $h$, cf.~\cite{BechererBilarevFrentrup2016-deterministic-liquidation}. There $M$ can be taken even quasi-left continuous in which case the set of admissible strategies should be restricted to predictable processes. 

In the general case, the following  compactness argument proves  existence of an optimizer - without providing any structural description for it, of course.
First, it suffices to optimize over deterministic strategies and thus to take $M \equiv 1$ by a change of measure argument, see \cite[Remark~3.9]{BechererBilarevFrentrup2016-deterministic-liquidation}. Now, for some fixed $\varepsilon > 0$ consider the optimization problem over the set of strategies
 $$\widetilde{\mathcal{A}}_{\text{mon}}(\theta) = \{\widetilde{\assetsProcess} \in D[-\varepsilon, T+\varepsilon] \mid \widetilde{\assetsProcess} \text{ is the \emph{extended path} of some determ. } \assetsProcess\in \admissibleSellStrategies{\theta}\}.$$
Endowing $\widetilde{\mathcal{A}}_{\text{mon}}(\theta)$ with the Skorokhod $M_1$ topology makes it relatively compact, which is straightforward to check using \cite[Thm.~12.12.2]{Whitt2002_book}; the compactness criterion in \cite[Thm.~12.12.2]{Whitt2002_book} is trivial for such monotone strategies because the $M_1$ oscillation function is zero and all the paths are constant in neighborhoods of the end points. Thus, if $(\widetilde{\assetsProcess^n})\subset \widetilde{\mathcal{A}}_{\text{mon}}(\theta)$ is a maximizing sequence (of extended paths) for the problem \eqref{eq:opt liq problem monotone}, then it (or some subsequence) converges to $\widetilde{\assetsProcess^*}\in D[-\varepsilon, T+\varepsilon]$. By continuity of the proceeds functional $L$ in the $M_1$ topology (\cref{thm:stability in j1 and m1}) we obtain  
\begin{equation}\label{eq:sup is attained}
	\sup_{\assetsProcess\in \admissibleSellStrategies{\theta}} L_T(\assetsProcess) = \lim_{n\rightarrow \infty} L_{T+\varepsilon}(\widetilde{\assetsProcess^n}) = L_{T+\varepsilon}(\widetilde{\assetsProcess^*}).
\end{equation}
Since on $[-\varepsilon,0)$ (resp.~$(T, \varepsilon]$) each $\widetilde{\assetsProcess^n}$ is constant $\theta$ (resp.~0) and convergence in $M_1$ implies local uniform convergence at continuity points of the limit, cf.~\cite[Lemma~12.5.1]{Whitt2002_book}, there exists $\assetsProcess^*\in \admissibleSellStrategies{\theta}$ such that $\widetilde{\assetsProcess^*}$ is its extended path in $D[-\varepsilon, T+\varepsilon]$. Thus $L_{T+\varepsilon}(\widetilde{\assetsProcess^*}) = L_T(\assetsProcess^*)$ and $\assetsProcess^*$ is an optimal liquidation strategy by \eqref{eq:sup is attained}.

\subsection{Optimal liquidation problem with general strategies} \label{ex: optimal liquidation with general strategies}
	Consider the problem from \cite[Sect.~5]{BechererBilarevFrentrup2016-deterministic-liquidation} to liquidate a risky asset optimally, posed over the set of bounded variation strategies $\admissibleFiniteVariationStrategies{\theta}$ with no shortselling, for some initial position $\theta\geq 0$, i.e.\ $\max_{\assetsProcess \in \admissibleFiniteVariationStrategies{\theta}} \EE[L_\infty(\assetsProcess)]$; Recall that in the setup there the fundamental price process is $\baseS_t = e^{-\delta t} M_t$ for some $\delta > 0$ and a non-negative locally square integrable quasi-left continuous martingale $M$, and $\diff \langle M \rangle_t $ in the dynamics of $Y$ in  \eqref{eq:deterministic Y_t dynamics} is replaced by $\diff t$.
	By \cite[Thm.~5.1]{BechererBilarevFrentrup2016-deterministic-liquidation}, the optimal bounded variation strategy $\assetsProcess^*$ is deterministic and liquidates in  some finite time $T-1$ (which depends on the model parameters).
	
	Now consider the optimal liquidation problem over the  larger set of admissible strategies
	\[
		\admissibleSemimartingaleStrategies(\theta):= \{ \assetsProcess \mid \text{bounded predictable semimartingale, } \assetsProcess \ge 0, \assetsProcess_{0-} = \theta, \assetsProcess_t = \assetsProcess_{t \wedge (T-1)} \}.
	\]
	Note that for any admissible strategy $\assetsProcess\in \admissibleSemimartingaleStrategies(\theta)$, the (martingale part of the) stochastic integral in \cref{eq:def of proceeds process} is a true martingale and will vanish in expectation, yielding
	\[
		\EE[ L_T(\assetsProcess)] = \EE\brackets[\bigg]{  - \int_0^T e^{-\delta t}M_t ( (fh)(Y^\assetsProcess_t) +\delta F(Y^\assetsProcess_t))\diff t - (e^{-\delta T}M_T F(Y^\assetsProcess_T) - M_{0-} F(Y^\assetsProcess_{0-}))},
	\]
	where $F(x) = \int_0^x f(y)\diff y$.
	A change of measure argument as in  \cite[Rem.~3.9]{BechererBilarevFrentrup2016-deterministic-liquidation} shows that we can take w.l.o.g.~$M\equiv 1$ and thus it suffices to optimize the proceeds over the set $\mathcal{A}_{\text{c\`{a}dl\`{a}g}}(\theta)$ of all deterministic non-negative c\`{a}dl\`{a}g paths having square-summable jumps, starting at time $0-$ at $\theta$ and   being zero after time $T-1$. For each such $\assetsProcess\in \mathcal{A}_{\text{c\`{a}dl\`{a}g}}(\theta)$ and every $\varepsilon > 0$, we can find a deterministic bounded variation strategy $\assetsProcess^\varepsilon\in \admissibleFiniteVariationStrategies{\theta}$ that executes until time $T$ and gives proceeds that are at most $\varepsilon$-away from the proceeds of $\assetsProcess$. Indeed, this follows from \cref{cor:WZ proceeds converge} where the approximating sequence is indeed of bounded variation continuous processes (since $\assetsProcess$ is bounded), and noting that the probabilistic nature of the stability results in \cref{sec:main stability results} is due to the presence of the (intrinsically probabilistic) stochastic integral in \eqref{eq:def of proceeds process}, cf.~the proof of \cref{thm:stability in j1 and m1}, which would be immaterial here in the case of constant $M$. In particular, 
		\[
		\sup_{\assetsProcess \in \admissibleSemimartingaleStrategies(\theta)} \EE[L_T(\assetsProcess)] 
		\leq \sup_{\mathcal{A}_{\text{c\`{a}dl\`{a}g}}(\theta)} \EE[L_T(\assetsProcess)] = \sup_{\assetsProcess \in \admissibleFiniteVariationStrategies{\theta}} \EE[L_T(\assetsProcess)] 
		= \EE[L_T(\assetsProcess^*)],
	\]
	meaning that $\assetsProcess^*$ is optimal also within in the (larger) set $\admissibleSemimartingaleStrategies(\theta)$.

\newcommand{\typeFname}{impact fixing}
\subsection{Stochastic liquidity and constrained liquidation horizon}
\label{ex: stochastic-finite-horizon}
\newcommand{\helperFunction}{\psi}
\newcommand{\helperFunctionHull}{\hat \helperFunction}
\newcommand{\otherHelperFunction}{\hat \Psi}

Let us investigate an optimal liquidation problem 
 for a variant of the price impact model which features \emph{stochastic liquidity}.
The singular control problem exhibits two interesting properties:  It still permits an explicit description for the optimal  strategy
under a new constraint on the \emph{expected} time to  (complete) liquidation,
but the optimal control is not of finite variation.
So the set of admissible strategies needs to accommodate for infinite variation controls.
As it is clear how to define the proceeds functional for (continuous) strategies of finite variation (cf.~\eqref{eq:proceeds cont proc}), and we want (and need) to admit for jumps in the (optimal) control, the $M_1$ topology is a natural choice to  extend the domain continuously.
We consider no discounting or drift in the unaffected price process, letting $\baseS_t = \baseS_0 \scE(\sigma W)_t$ with constant $\sigma > 0$. This martingale case will permit to apply convexity arguments in spirit of  \cite{PredoiuShaikhetShreve11} to construct an optimal control, see \cref{thm:stoch resilience} below.
In \eqref{eq:deterministic Y_t dynamics}, the dynamics of market impact  $Y$ (called volume effect in~\cite{PredoiuShaikhetShreve11})  was a deterministic function of the large trader's strategy $\assetsProcess$. 
To model liquidity which is stochastic (e.g.\ by volume imbalances from other large 'noise' traders, cf.\ 
\cite[Remark~2.4]{BechererBilarevFrentrup2016-stochastic-resilience}), 
 let the  impact process  $Y^\assetsProcess$ solve
\begin{equation}
	\diff Y^\assetsProcess_t = -\beta Y^\assetsProcess_t \diff t + \impactVolatility \diff B_t + \diff \assetsProcess_t\,, \quad  \text{ with }\quad Y^\assetsProcess_{0-} = Y_{0-} \in \RR \text{ given,}
\end{equation}
for constants $\beta, \impactVolatility > 0$ and a Brownian motion $B$ that is independent of $W$.
For the impact function $f \in C^3(\RR)$, giving the observed price  by $S_t = f(Y_t) \baseS_t$, we require $f,f' > 0$ with $f(0)=1$ and that $\lambda(y) := f'(y)/f(y)$ is bounded away from $0$ and $\infty$, i.e.\ for constants $0<\lambda_{\min}\le \lambda_{\max}$ we have $\lambda_{\min} \le \lambda(y) \le \lambda_{\max}$ for all $y \in \RR$, with bounded derivative $\lambda'$.
Moreover, we assume that $k(y):= \frac{\impactVolatility^2}{2} \frac{f''(y)}{f(y)} - \beta - \beta y \frac{f'(y)}{f(y)}$ is strictly decreasing.
An example satisfying these conditions is $f(y) = e^{\lambda y}$ with constant $\lambda > 0$.
Let $F(x):= \int_{-\infty}^x f(y) \diff y$, which is positive and of exponential growth due to the bounds on $\lambda$: $0 < F(x) \le \bp{ e^{\lambda_{\min}} + e^{\lambda_{\max}} } / \lambda_{\min}$.
The liquidation problem on infinite horizon \emph{with} discounting and \emph{without} intermediate buying in this model has been solved in \cite{BechererBilarevFrentrup2016-stochastic-resilience}.

For our problem here,
proceeds of general semimartingale strategies $\assetsProcess$ should be
\begin{align} \label{def:stochastic-impact-proceeds}
\begin{split}
	L_T(\assetsProcess) 
		&= \int_0^T \baseS_t \helperFunction(Y^\assetsProcess_{t-}) \diff t + \baseS_0 F(Y_{0-}) - \baseS_T F(Y^\assetsProcess_T)
\\		&\qquad+ \int_0^T F(Y^\assetsProcess_{t-}) \diff \baseS_t + \impactVolatility \int_0^T \baseS_t f(Y^\assetsProcess_{t-}) \diff B_t 
		\,,
\end{split}
\end{align}
with $\helperFunction(y):= -\beta y f(y) + \frac{\impactVolatility^2}{2} f'(y)$,
because \eqref{def:stochastic-impact-proceeds} is the continuous extension (in $M_1$ in probability, as in \cref{thm:stability in j1 and m1}) of the functional $L(\assetsProcess^c) = -\int_0^T S_u \diff \assetsProcess^c_u$ from continuous f.v.~$\assetsProcess^c$ to   semimartingales $\assetsProcess$ that are \emph{bounded in probability}  on $[0,\infty)$:
The proof of \cref{thm:stability in j1 and m1} carries over as for such $\assetsProcess$, impact $Y$ and thus $\helperFunction(Y)$ and $F(Y)$ are then also bounded in probability and the stochastic $\diff B$-integral in \eqref{def:stochastic-impact-proceeds} converges by a similar argument as in \eqref{conv of stoch integr} for the $\diff \baseS$ integral, using $\angles{\baseS}_t = \sigma^2 \int_0^t \baseS_u^2 \diff u = \sigma^2 \angles{\int_0^\cdot \baseS_u \diff B_u}_t$.

Our goal is to maximize expected proceeds $\EE[L_\infty(\assetsProcess)]$ over some suitable set of admissible strategies that we specify now.
From an application point of view, it makes sense to impose some bound on the time horizon within which liquidation is to be completed.
Indeed, since our control objective here involves no discounting, one needs to restrict the horizon to get 
a non-trivial solution.
Let some $\maxET \ge 0$ be given.
A semimartingale $\assetsProcess$ that is bounded in probability on $[0,\infty)$
will be called an \emph{admissible strategy}, if 
\begin{align*}
&\text{there exists a stopping time $\tau$ with $\EE[\tau] \le \maxET$ such that $\assetsProcess_t = \assetsProcess_t \indicator_{t \le \tau}$, with}\\
&\text{$\EE[\tau \baseS_\tau] < \infty$, $\bp{L_\tau(\assetsProcess)}^- \in L^1(\PP)$ and such that the processes $\int_0^{\cdot \wedge \tau} \baseS_t F(Y^\assetsProcess_{t-}) \diff W_t$\,,}\\
&\text{$\int_0^{\cdot \wedge \tau} \baseS_t f(Y^\assetsProcess_{t-}) \diff B_t$\,, \  $\baseS_{\cdot \wedge \tau}$ and $(\baseS B)_{\cdot \wedge \tau}$ are uniformly integrable (UI) martingales. }
\end{align*}
The integrability conditions ensure $L_\tau(\assetsProcess) \in L^1(\PP)$. Indeed, for admissible $\assetsProcess$ it suffices to check $\bp{\int_0^\tau \baseS_t \helperFunction(Y^\assetsProcess_{t-}) \diff t}^+ \in L^1(\PP)$.
We will show in the proof of \cref{thm:stoch resilience} that $\helperFunction$ attains a maximum $\helperFunction(y^*)$. Thus we can bound $\int_0^\tau \baseS_t \helperFunction(Y^\assetsProcess_{t-}) \diff t$ from above by $\helperFunction(y^*) \int_0^\tau \baseS_t \diff t$, which is integrable by optional projection \cite[Thm.~VI.57]{DellacherieMeyer82bookB} since $\EE[\tau \baseS_\tau] < \infty$.

Let $\scA_\maxET$ be the set of all admissible strategies with given fixed initial value $\assetsProcess_{0-}$, where $\abs{\assetsProcess_{0-}}$ is the number of shares to be liquidated (sold) if $\assetsProcess_{0-} >0$, resp.\ acquired (bought) if $\assetsProcess_{0-} <0$. 
The definition of  $\scA_\maxET$ involves several technical conditions. But  the set $\scA_\maxET$ is not small, for instance it clearly contains all strategies of finite variation which liquidate until some bounded stopping times $\tau$ with $\EE[\tau]\leq \maxET$, and also 
strategies of infinite variation (see below). 
Note that intermediate short selling is permitted, and that $\scA_0$ contains only the trivial strategy to sell (resp.~buy) everything immediately.

 We will
show that optimal strategies are  \emph{\typeFname{}}.
For $\tilde \Upsilon, \Upsilon \in \RR$ an \emph{\typeFname{} strategy} $\assetsProcess = \assetsProcess^{\tilde\Upsilon,\Upsilon}$ is a strategy with liquidation time $\tau$ (i.e.~$\assetsProcess_t = 0$ for $t\geq \tau$), such that $Y = Y^{\assetsProcess^{\tilde\Upsilon,\Upsilon}}$ satisfies $Y_t = \tilde \Upsilon$ on $\rightOpenStochasticInterval{0,\tau}$ and $Y_\tau = \Upsilon$. More precisely, $\assetsProcess_0 = \assetsProcess_{0-} + \tilde\Upsilon - Y_{0-}$,  $\diff \assetsProcess_t = \beta \tilde\Upsilon\diff t - \impactVolatility \diff B_t$ on $\openStochasticInterval{0,\tau}$ until $\tau = \tau^{\tilde\Upsilon,\Upsilon} := \inf \{ t > 0 \mid \assetsProcess_{t-} = \tilde \Upsilon - \Upsilon \}$, with final block trade of size $\Delta \assetsProcess_\tau = -\assetsProcess_{\tau-} = \Upsilon - \tilde\Upsilon$ and $\assetsProcess=0$ on $\rightOpenStochasticInterval{\tau,\infty}$. 
We have the following properties of \typeFname{} strategies (for proof, see \cref{sect:finite stochastic horizon proofs}).
\begin{lemma}[Admissibility of \typeFname{} strategies]\label{lem:admissible type F strategies}
	The liquidation time $\tau = \tau^{\tilde\Upsilon,\Upsilon}$ of an \typeFname{} strategy $\assetsProcess^{\tilde\Upsilon,\Upsilon}$ has expectation 
	\(
		\EE[\tau] = \paren{Y_{0-} - \assetsProcess_{0-} - \Upsilon}/\paren{\beta \tilde\Upsilon}
	\)
	if $(Y_{0-} - \assetsProcess_{0-} - \Upsilon)\tilde\Upsilon > 0$, and $\EE[\tau] = 0$ if $\Upsilon = Y_{0-} - \assetsProcess_{0-}$, otherwise $\EE[\tau] = \infty$.

	\noindent
	Moreover, if $\EE[\tau^{\tilde\Upsilon,\Upsilon}] \le \maxET$ then $\assetsProcess^{\tilde\Upsilon,\Upsilon} \in \scA_\maxET$.
\end{lemma}

 Using convexity arguments we construct the solution for the optimization problem in 
\begin{theorem}\label{thm:stoch resilience}
	For every $\maxET \in [0,\infty)$ there exist $\hat\eta \in [0,\maxET]$ and $\tilde \Upsilon$, $\Upsilon \in \RR$ such that the  associated {\typeFname{} strategy} $\hat \assetsProcess := \assetsProcess^{\tilde \Upsilon, \Upsilon}$ generates maximal expected proceeds in expected time $\EE[\tau^{\tilde \Upsilon, \Upsilon}] = \hat\eta$ among all admissible strategies, i.e.
	\begin{equation*}
		\EE[L_\infty(\hat\assetsProcess)] = \max \braces[\big]{ \EE[L_\infty(\assetsProcess)]  \bigm| \assetsProcess \in \scA_\maxET}\,.
	\end{equation*}
	Moreover, if $f(y) = e^{\lambda y}$ with $\lambda\in (0,\infty)$, then we have $\hat\eta = \maxET$ and the optimal strategy is unique.
\end{theorem}
The proof will also show that optimal strategies have to be  \emph{\typeFname{}}. In particular, 
any non-trivial  admissible strategy of finite variation is suboptimal.
\begin{proof}
Since  $f'/f$ and $(f'/f)'$ 
  are bounded, then $f''/f$ is also bounded and hence there is a unique $y^* \in \RR$ with $k(y^*) = 0$.
So $\helperFunction$ is strictly increasing on $(-\infty,y^*)$ and decreasing on $(y^*,\infty)$, since $\helperFunction'(y) = f(y)k(y)$.
Note that $\helperFunction$ is strictly concave on $[y^*, \infty)$ and $\helperFunction(y) > 0$ for $y < 0$. Hence, the concave hull of $\helperFunction$ is
\[
	\helperFunctionHull(y):= \inf \{ \ell(y) \mid \text{$\ell$ is an affine function with $\ell(x) \ge \helperFunction(x)\ \forall x$} \} = \helperFunction(y \vee y^*)\,.
\]
Let $\assetsProcess \in \scA_\maxET$ with liquidation time $\tau$.
Denote by $\QQ$ the measure with $\diff \QQ = \bp{ \baseS_\tau / \baseS_0 } \diff \PP$.
Then by optional projection, as in \cite[Thm.~VI.57]{DellacherieMeyer82bookB}, we obtain (taking w.l.o.g.~$\baseS_0 = 1$):
\begin{align}
\nonumber	\EE[L_\infty] 
		&= \EE[L_\tau] 
		= \EE\brackets[\Big]{ \int_0^\tau \baseS_t \helperFunction(Y_t) \diff t } + F(Y_{0-}) - \EE\brackets[\big]{ \baseS_{\tau} F(Y_\tau) }
\\
\nonumber		&= F(Y_{0-}) + \EE_\QQ\brackets[\Big]{ \int_0^\tau \helperFunction(Y_t) \diff t } - \EE_\QQ\brackets[\big]{ F(Y_\tau) }
\\ \label{eq:expected proceeds via measure mu}
		&= F(Y_{0-}) + \int_{\Omega \times [0,\infty)} \helperFunction(Y_t(\omega)) \mu(\!\diff \omega, \!\diff t) - \EE_\QQ\brackets[\big]{ F(Y_\tau) }
\,,
\end{align}
for the finite measure $\mu$ given by $\mu(A \times B):= \int_A \int_0^{\tau(\omega)} \indicator_B(t) \diff t \QQ[\!\diff \omega]$ with total mass $\mu(\Omega,[0,\infty)) = \EE_\QQ[\tau] = \EE[\tau \baseS_\tau] < \infty$.
For $\tau \neq 0$, Jensen's inequality for $\helperFunctionHull$ and $F$ gives
\begin{align}
\label{ineq:concave hull above helperFunction}
	\EE[L_\infty] &\le F(Y_{0-}) + \int_{\Omega \times [0,\infty)} \helperFunctionHull(Y_t(\omega)) \mu(\!\diff \omega, \!\diff t) - \EE_\QQ\brackets[\big]{ F(Y_\tau) }
\\ \label{ineq:Jensen for concave hull of helperFunction}
		&\le F(Y_{0-}) + \EE_\QQ[\tau] \helperFunctionHull\paren[\bigg]{ \frac{1}{\EE_\QQ[\tau]} \int_{\Omega \times [0,\infty)} Y_t(\omega) \mu(\!\diff\omega,\!\diff t) } - \EE_\QQ\brackets[\big]{ F(Y_\tau) }
\allowdisplaybreaks
\\ \label{eq: E_Q of int beta Y inside helperFunction hat}
		&= F(Y_{0-}) + \EE_\QQ[\tau] \helperFunctionHull\paren[\bigg]{ \frac{1}{\beta \EE_\QQ[\tau]} \EE_\QQ\brackets[\Big]{ \int_0^\tau \beta Y_t \diff t } } - \EE_\QQ\brackets[\big]{ F(Y_\tau) }
\allowdisplaybreaks
\\ \label{eq: E_Q of Y inside helperFunction hat}
		&= F(Y_{0-}) + \EE_\QQ[\tau] \helperFunctionHull\paren[\bigg]{ \frac{Y_{0-} - \assetsProcess_{0-} - \EE_\QQ[Y_\tau]}{\beta \EE_\QQ[\tau]} } - \EE_\QQ\brackets[\big]{ F(Y_\tau) }
\allowdisplaybreaks
\\ \label{ineq:Jensen for F}
		&\le F(Y_{0-}) + \EE_\QQ[\tau] \helperFunctionHull\paren[\bigg]{ \frac{Y_{0-} - \assetsProcess_{0-} - \EE_\QQ[Y_\tau]}{\beta \EE_\QQ[\tau]} } - F\bp{\EE_\QQ[Y_\tau] }
\\ \label{eq:expected proceeds bound fct of Etau and EYtau}
		&= F(Y_{0-}) + \otherHelperFunction\bp{ \EE_\QQ[\tau], \EE_\QQ[Y_\tau] }
\,,
\end{align}
for $\otherHelperFunction(\eta, \Upsilon):= \eta \helperFunctionHull\bp{ \frac{Y_{0-} - \assetsProcess_{0-} - \Upsilon}{\beta \eta} } - F(\Upsilon)$ when $\eta > 0$, while for $\tau = 0$ we get that $\EE[L_\infty]$ is given by \eqref{eq:expected proceeds bound fct of Etau and EYtau} with  $\otherHelperFunction(0, \Upsilon):= - F(\Upsilon)$.
The step from \eqref{eq: E_Q of int beta Y inside helperFunction hat} to \eqref{eq: E_Q of Y inside helperFunction hat} uses that $\EE[ \baseS_\tau B_\tau ] = 0$, due to $(\baseS B)_{\cdot\wedge\tau}$ being UI, and $\int_0^t \beta Y_s \diff s = \impactVolatility B_t + \assetsProcess_t - \assetsProcess_{0-} - Y_t + Y_{0-}$.
Since $F$ is strictly convex, we obtain equality in \eqref{ineq:Jensen for F} if and only if $Y_\tau$ is concentrated at a point $\Upsilon \in \RR$ $\PP$-a.s.
At \eqref{ineq:Jensen for concave hull of helperFunction} we obtain equality if and only if either $Y_t \in (-\infty, y^*]$ $\mu$-a.e. (where $\helperFunctionHull$ is affine) or $Y_t$ is concentrated at a point $\tilde\Upsilon \in \RR$ $\mu$-a.e.
Equality at \eqref{ineq:concave hull above helperFunction} can only happen if $Y\geq y^*$ $\mu$-a.e.
Hence, we only get equality
\begin{equation*}
	\EE[L_\infty] = F(Y_{0-}) + \otherHelperFunction\bp{ \EE_\QQ[\tau], \EE_\QQ[Y_\tau] }
\end{equation*}
for \emph{\typeFname{} strategies} $\assetsProcess = \assetsProcess^{\tilde\Upsilon,\Upsilon}$ with $\tilde \Upsilon \geq y^*$,
where $\EE[L_\tau] = F(Y_{0-}) + \otherHelperFunction(\EE[\tau], \Upsilon)$.
Since $y^*$ is the largest maximizer of $\helperFunctionHull$, $\lim_{y \to \infty} \helperFunction'(y) = -\infty$ and $F$ is strictly increasing, $\otherHelperFunction(\eta, \cdot)$ has a unique maximizer $\hat e(\eta) \in (-\infty, e^*)$ where $e^* = e^*(\eta) = Y_{0-} - \assetsProcess_{0-} - \beta \eta y^*$ for $\eta > 0$ and $\hat e(0) = e^*(0) = Y_{0-} - \assetsProcess_{0-}$. 
Because $\hat y(\eta):= \bp{ Y_{0-} - \assetsProcess_{0-} - \hat e(\eta) }/(\beta\eta) > y^*$, the {\typeFname{} strategy} $\assetsProcess^{\hat y(\eta), \hat e(\eta)}$ has expected time to liquidation $\eta$ (cf.~\cref{lem:admissible type F strategies}) and generates $F(Y_{0-}) + \otherHelperFunction\bp{ \eta, \hat e(\eta) }$ expected proceeds that are optimal among all {\typeFname{} strategies} with expected time to liquidation $\eta$.

Note that $\hat{e}(\eta)$ is continuous in $\eta\in (0,+\infty)$ by the implicit function theorem; recall that $\hat{e}(\eta)$ solves $0=\otherHelperFunction_\Upsilon(\eta, \hat{e}(\eta)) = -\helperFunctionHull'(\hat{y}(\eta))/\beta - f(\hat{e}(\eta))$, and $\otherHelperFunction_{\Upsilon\Upsilon}(\eta, \Upsilon) < 0$ for $\Upsilon < e^*(\eta)$. Moreover, $\hat{e}(\eta)\to \hat{e}(0)$ when $\eta\to 0$, otherwise $\hat{y}(\eta)\to +\infty$ for a subsequence giving $-\helperFunctionHull'(\hat{y}(\eta))/\beta =-(fk)(\hat{y}(\eta))/\beta \to +\infty $ and therefore also $f(\hat{e}(\eta))\to +\infty$, which would contradict $\limsup_{\eta\to 0} \hat{e}(\eta) \leq \lim_{\eta\to 0}e^{*}(\eta) = Y_{0-}-\assetsProcess_{0-}$. 

In particular, the contradiction argument above shows that $\hat{y}(\eta)$ is contained in a compact set for small $\eta$. As a consequence, $\hat{G}(\eta, \hat{e}(\eta)) = \eta \helperFunctionHull(\hat{y}(\eta)) - F(\hat{e}(\eta)) \to \hat{G}(0, \hat{e}(0))$ as $\eta\to 0$, i.e.\ the map $\eta\mapsto \otherHelperFunction(\eta, \hat{e}(\eta))$ is continuous on $[0,+\infty)$.
 Hence,  it attains a maximizer $\hat\eta \in [0,\maxET]$ whose associated {\typeFname{} strategy} $\hat \assetsProcess = \assetsProcess^{\hat y(\hat \eta), \hat e(\hat \eta)}$ generates maximal expected proceeds in expected time $\EE[\tau^{\hat y(\hat \eta), \hat e(\hat \eta)}] = \hat\eta$ among all admissible strategies $\scA_\maxET$.
 
If $f(y) = e^{\lambda y}$ with  $\lambda\in (0,\infty)$, one can check by direct calculations that $\hat{G}_\eta(\eta, \Upsilon)  > 0$ for $\eta>0$, $\Upsilon \in \RR$, and thus using $\frac{\mathrm d}{\mathrm d \eta} \hat{G}(\eta, \hat{e}(\eta)) = \hat{G}_\eta(\eta, \hat{e}(\eta)) +  \hat{G}_{\Upsilon}(\eta, \hat{e}(\eta)) \hat{e}'(\eta) = \hat{G}_\eta(\eta, \hat{e}(\eta))$, the map $\eta \mapsto \otherHelperFunction(\eta,\hat{e}(\eta))$ is strictly increasing, so $\hat\eta = \maxET$ is its unique maximizer in $[0,\maxET]$ and hence the optimal strategy is unique.
\end{proof}

\subsection{Price impact with partially instantaneous recovery} \label{ex: partial instantaneous impact}
This example is inspired by work of \cite{Roch11} on a different (additive impact, block-shaped limit order book (LOB))
 price impact model; adapting his interesting idea to our setup leads to an extension of our
 transient impact 
model,
where a further parameter $\eta\in (0,1]$ permits for partially instantaneous recovery of price impact.
Further, the example illustrates how 
proceeds from 
trading could, at first, be given for simple strategies only, and 
continuity arguments are key for an extension to a larger space of strategies.

Motivated by 
observations
that other traders 
respond quickly to 
market orders by adding limit orders in opposite direction, \cite{Roch11} has proposed a model where impact from a block trade is partially instantaneous and partially 
transient.
A market sell (resp.\ buy) order
eats into the bid (resp.\ ask) side of a LOB and is filled at respective prices, price impact being a function of the shape of the LOB.
A certain  fraction $1-\eta$ ($0< \eta\le 1$) of that
impact 
is instantaneously recovered directly after the trade, while only the remaining $\eta$-fraction constitutes a transient impact that decays gradually over time (cf.\ \eqref{eq: impact with eta}). 
As stated in \cite{Roch11}, this means that ``we think of $1-\eta$ as the fraction of the order book which is renewed after a market order so that in practice the actual impact on prices is $\eta$
 times the full impact''.
In our previous model for a two-sided LOB (non-monotone strategies),  with the idealizing assumption of zero bid-ask spread, 
the model with full impact ($\eta=1$) implicitly postulates 
that the gap between bid and ask prices after a block buy (resp.~sell) order is 
filled up instantaneously with ask (resp.~bid) orders. For one-directional trading such hypothesis is conservative, but for trading in alternating directions it may be overly optimistic.
So,  it appears 
to be an interesting generalization to postulate that the gap is closed from both sides in a certain fraction.

To incorporate this into our setup, let $\eta \in [0,1]$ and  suppose that the impact directly after completion of a block trade of size $\Delta\assetsProcess_t$ at time $t\in [0,\infty)$ is actually $Y_{t-} + \eta\Delta\assetsProcess_t$, where $Y_{t-}$ is the market impact immediately before the trade. 
Thus, the market impact process $Y^{\eta, \assetsProcess}$  evolves according to 
\begin{equation}\label{eq: impact with eta}
\diff Y^{\eta, \assetsProcess}_t = -h(Y^{\eta, \assetsProcess}_t)\diff \langle M\rangle_t + \eta \diff \assetsProcess_t, \quad t\geq 0.
\end{equation}
Indeed, \eqref{eq: impact with eta} holds for simple strategies $\assetsProcess$ and hence for all c\`{a}dl\`{a}g trading programs $\assetsProcess$ by continuity of $\assetsProcess\mapsto Y^{\eta, \assetsProcess}$ in the  uniform and Skorokhod $J_1$ and $M_1$ topologies. 

The case $\eta = 0$ corresponds to 
no (non-instantaneous) impact while $\eta = 1$ gives our previous setup with full impact.
The situation where $\eta\in (0, 1)$ is more delicate, in that executing a block order at once would always be suboptimal, whereas subdividing a block trade into smaller ones and executing them one after the other would lead to smaller expenses, i.e.\ larger proceeds, due to the instantaneous partial recovery 
of price impact.
Thus, there would be a difference between \textit{asymptotically realizable} proceeds from a block trade (in the terminology of \cite{BankBaum04}) and its direct proceeds from a LOB interpretation. 

Motivated by optimization questions like the optimal trade execution problem where a trader tries to evade illiquidity costs from large (block) orders, if possible, our aim is to specify a model that is stable with respect to small intertemporal changes, in particular approximating block trades by subdividing the trade into small packages and executing them in short time intervals. 
Thus, the proceeds that we will derive here will be asymptotically realizable. First, let us only assume that at every time $t\geq 0$, the average price per share for a block trade of size $\Delta$ is some value between $f(Y_{t-})\baseS_t$ and $f(Y_{t-} + \Delta)\baseS_t$, where $Y_{t-}$ is the state of the impact process right before the block trade. 
Hence, the arguments in the proof of \cref{lemma:proceeds of cont fv strategies} carry over (with $c=1/\eta$, $Y=Y^\eta/\eta$ and suitably re-scaled functions $f$,$h$)
and yield that the proceeds from implementing a continuous finite variation strategy $\assetsProcess$ should be given by
\({\small
	\tilde L_T(\assetsProcess) = -\int_0^T \baseS_t f(Y^{\eta, \assetsProcess}_{t})\diff \assetsProcess_t,
}
\) $T\ge 0$, irrespective of a particular initial specification for proceeds from block trades.
As such was the starting point for \cref{sect:continuity}, the analysis there for the case $\eta=1$ carries over 
to the model extension for $\eta\in(0,1]$: For any continuous f.v.\ process $\assetsProcess$ we obtain
\begin{equation}\label{eq:proceeds partial impact}
	\hspace{-\mathindent}
	\squeeze{
	\tilde L_T(\assetsProcess) 
	\!
		=\! \frac{1}{\eta}\paren[\bigg]{ 
			\int_0^T \!\!\!\! F(Y^{\eta,\assetsProcess}_{u-})\diff \baseS_u 
			-\! \int_0^T \!\!\!\! \baseS_u (fh)(Y^{\eta,\assetsProcess}_u)\diff \langle M \rangle_u 
			- \paren[\big]{ \baseS_T F(Y^{\eta,\assetsProcess}_T) - \baseS_0 F(Y^{\eta,\assetsProcess}_{0-}) } 
		\!}
	}
	\mathrlap{.}
\end{equation}
By \cref{thm:stability in j1 and m1} the right-hand side of \eqref{eq:proceeds partial impact} is continuous in the predictable strategy $\assetsProcess$ taking values in $D([0,T];\RR)$ when endowed with any of the uniform, Skorokhod $J_1$ and $M_1$ topologies. 
So, asymptotically realizable proceeds are given by \eqref{eq:proceeds partial impact}. 
In particular, asymptotically realizable proceeds from a block sale of size $\Delta\neq 0$ at time $t$ are 
\[
	-\frac{1}{\eta} \baseS_t \paren[\big]{ F(y_{t-} + \eta \Delta) - F(y_{t-}) } 
		= -\frac 1 \eta \baseS_t\int_0^{\eta \Delta} f(y_{t-} + x)\diff x
	\,,
\]
where $y_{t-}$ denotes the state of the market impact process before the trade. Note that these proceeds strictly dominate
the proceeds  $- \baseS_t\int_0^{ \Delta} f(y_{t-} + x)\diff x$  that would arise from a executing the block sale in the  LOB 
corresponding to the price impact function $f$. 
Also this model variant is free of arbitrage in the sense of \cref{thm:absence of arbitrage}, whose proof carries over.
In mathematical terms one may observe, maybe surprisingly, that the model structure (see \eqref{eq: impact with eta} and \eqref{eq:proceeds partial impact}) for the extension $\eta\in (0,1]$ is like the one for the previous model (with $\eta=1$), and is hence amenable to a likewise analysis. In finance terms,  to model partially instantaneous recovery in such a way thus has  quantitative effects. But it does not lead to new qualitative features for the model, since the large investor could side-step much of the, at first sight, highly disadvantageous effect from large block trades by trading continuously (in approximation), at least in absence of further frictions.

\appendix
\section{Appendix}

The next proposition collects known continuity properties 
of the solution map $\assetsProcess\mapsto Y^\assetsProcess$ on $D([0,T];\RR)$  from \eqref{eq:deterministic Y_t dynamics}, with the presentation being adapted to our setup.
\begin{proposition}\label{prop:cont of resilience}
Assume that $h$ is Lipschitz continuous  and $\langle M \rangle = \int_0^\cdot \alpha_s \diff s$ with pathwise (locally) Lipschitz density $\alpha$.
Then the solution map $D([0,T];\RR)\to D([0,T];\RR)$, with $\assetsProcess\mapsto Y^\assetsProcess$ from \eqref{eq:deterministic Y_t dynamics}, is  defined pathwise. The map is continuous when the space $D([0,T];\RR)$ is endowed with either the uniform topology or the Skorokhod $J_1$ or $M_1$ topology. Moreover, if $\assetsProcess$ is an adapted c\`{a}dl\`{a}g process, then the process $Y^\assetsProcess$ is also adapted.
\end{proposition} 
\begin{proof}
The proof in the case of the uniform topology and the Skorokhod $J_1$ topology is given in \cite[proof of Thm.~4.1]{PangTalrejaWhitt2007}; the proof there is for $\alpha\equiv 1$ but it clearly extends to our setup as long as $\alpha$ is Lipschitz. 
For the $M_1$ topology, cf.~\cite[Thm.~1.1]{PangWhitt2010}, where again the main argument (\cite[proof of Thm.~1.1]{PangWhitt2010}) extends to our setup of more general $\alpha$.  
That $Y^\assetsProcess$ is adapted follows from the (pathwise) construction of $Y^\assetsProcess$ as the (a.s.) limit (in the uniform topology) of adapted processes, the solution processes for a sequence of piecewise-constant controls $\assetsProcess^n$ approximating uniformly $\assetsProcess$, cf.~ \cite[proof of Thm.~4.1]{PangTalrejaWhitt2007}.
\end{proof}

In general, we may have $\alpha_n\to \alpha$ and $\beta_n\to \beta$ in $D([0,T])$ endowed with $J_1$ (or $M_1$), and yet  $\alpha_n+\beta_n\not\to \alpha+\beta$ when $\alpha$ and $\beta$ have a common jump time. However, in special cases like in what follows, this does not happen.
\begin{lemma}[Allowed cancellation of jumps for $J_1$]\label{lemma: J1 cancellation}
Let $\alpha_n \rightarrow \alpha_0$ and $\beta_n \rightarrow \beta_0$ in $(D([0,T]), J_1)$ with the following property: for every $n\geq 0$ and every $t\in (0,T)$ 
\begin{center}
	$\Delta \alpha_n(t) \neq 0 \quad  \text{implies}\quad \Delta \beta_n(t) = -\Delta \alpha_n(t).$
\end{center}
Then $\alpha_n + \beta_n \rightarrow \alpha_0 + \beta_0$ in $(D([0,T]), J_1)$.
\end{lemma}
\begin{proof}
By \cite[Prop.~VI.2.2, a]{JacodShiryaev2003_book} it suffices to check that for every $t\in (0,T)$ there exists a sequence $t_n\to t$ such that $\Delta \alpha_n(t_n)\to \Delta \alpha_0(t)$ and $\Delta \beta_n(t_n)\to \Delta \beta_0(t)$.

Let $t\in (0,T)$ be arbitrary and first suppose that $\Delta \alpha_0(t) \neq 0$. Then \cite[Prop.~VI.2.1, a]{JacodShiryaev2003_book} implies the existence of a sequence $t_n\to t$ such that $\Delta \alpha_n(t_n)\to \Delta \alpha_0(t)$. Thus, our assumption on the sequence $(\beta_n)$ gives  $\Delta \beta_n(t_n)\to \Delta \beta_0(t)$. For the case $\Delta \alpha_0(t) = 0$, let $t_n\to t$ be such that $\Delta \beta_n(t_n)\to \Delta \beta_0(t)$. By \cite[Prop.~VI.2.1, b.5]{JacodShiryaev2003_book} we conclude that $\Delta \alpha_n(t_n)\to \Delta \alpha_0(t)$ as well, finishing the proof.
\end{proof}
Let us note that the conclusion of \cref{lemma: J1 cancellation} does not hold for the $M_1$ topology. Consider for example $\alpha_0 = \indicator_{[1,\infty)}$  with approximating sequence $\alpha_n(t):= n\int_{t}^{t+1/n}\alpha_0(s)\diff s$ and $\beta_0 = 1 - \alpha_0$ with approximating sequence $\beta_n(t):= n\int_{t-1/n}^{t}\beta_0(s)\diff s$. Thus we need the following refined statement.

\begin{lemma}[Allowed cancellation of jumps for $M_1$]\label{lemma: M1 cancellation}
Let $\alpha_n \rightarrow \alpha_0$ in $(D([0,T]), \norm{\cdot}_\infty)$ and $\beta_n \rightarrow \beta_0$ in $(D([0,T]), M_1)$ with the following property: $t\in \text{Disc}(\alpha_0)$ implies 
$\beta_n\to\beta_0$ locally uniformly in a neighborhood of $t$.
Then $\alpha_n + \beta_n \rightarrow \alpha_0 + \beta_0$ in $(D([0,T]), M_1)$.
\end{lemma}
\begin{proof}
We prove the following claim that suffices to deduce $M_1$-convergence of $\alpha_n + \beta_n$: For any $t\in [0,T]$ and $\varepsilon> 0$ there are $\delta > 0 $ and $n_0\in \mathbb{N}$ such that
\begin{equation}\label{eq:osc common jumps}
	w_s(\alpha_n + \beta_n, t, \delta)\leq w_s(\alpha_n, t, \delta) + w_s(\beta_n, t, \delta) + \varepsilon \quad \text{for all $n\geq n_0$},
\end{equation}
where $w_s$ is the $M_1$ oscillation function, see \cite[Chap.~12, eq.~(4.4)]{Whitt2002_book}.
Indeed, if $\eqref{eq:osc common jumps}$ holds, then the second condition in \cite[Thm.~12.5.1(v)]{Whitt2002_book} would hold, while the first condition there holds because of local uniform convergence at points of continuity of $\alpha_0 + \beta_0$: Either there is cancellation of jumps and thus local uniform convergence by our assumption, or both paths do not jump which still gives local uniform convergence because $M_1$-convergence implies such at continuity points of the limit.

To check \eqref{eq:osc common jumps}, we have $\lim_{\delta \downarrow 0}\limsup_{n\rightarrow\infty}v(\alpha_n, \alpha_0, t, \delta) = 0$ at points $t\in [0,T]$ with $\Delta \alpha_0(t) = 0$, where for $x_1, x_2\in D([0,T])$
\[
	v(x_1, x_2, t, \delta):= \sup_{0\vee (t-\delta)\leq t_1,t_2 \leq (t+\delta)\wedge T}|x_1(t_1) - x_2(t_2)|,
\] 
see \cite[Thm.~12.4.1]{Whitt2002_book}, which implies \eqref{eq:osc common jumps} for small $\delta$ and large $n$. Now if $t\in \text{Disc}(\alpha_0)$, $\alpha_n\to \alpha_0$ and $\beta_n\to\beta_0$ locally uniformly in a neighborhood of $t$ which implies that for small $\delta$ and large $n$
\[
	w_s(\alpha_n + \beta_n, t, \delta)\leq w_s(\alpha_0 + \beta_0, t, \delta) + \varepsilon/2.
\]
Because $\alpha_0 + \beta_0\in D([0,T])$, we can make $w_s(\alpha_0 + \beta_0, t, \delta)$ smaller than $\varepsilon/2$, which finishes the proof.
\end{proof}

\begin{lemma}[Uniform convergence of jump term] \label{lemma: uniform jump-sum convergence}
	Let $\alpha, \beta_n, \beta \in D([0,T])$ be such that $[\alpha]^d_T:= \sum_{t\le T: \Delta\alpha(t) \ne 0} \abs{\Delta\alpha(t)}^2 < \infty$, $\beta_n$ are uniformly bounded and at every jump time $t \in [0,T]$ of $\alpha$, $\Delta \alpha(t) \ne 0$, we have pointwise convergence $\beta_n(t) \to \beta(t)$.
	Let $G \in C^2$ such that $y\mapsto G_{xx}(x,y)$ is Lipschitz continuous on compacts.
	Then the sum 
	\[
		J(\alpha,\beta_n)_t:= \sum_{\substack{u\le t \\ \Delta\alpha(t)\ne 0}} G\bp{\alpha(t), \beta_n(t)} - G\bp{ \alpha(t-), \beta_n(t) } - G_x\bp{\alpha(t-), \beta_n(t)} \Delta\alpha(t)
	\]
	converges uniformly for $t \in [0,T]$ to $J(\alpha,\beta)_t$, as $n \to \infty$.
\end{lemma}
\begin{proof}
\newcommand{\Lip}{L}
	Since $\alpha$, $[\alpha]^d$, $\beta_n$ and $\beta$ are bounded on $[0,T]$ by a constant $C \in \RR$, we can assume w.l.o.g.\ that $G_{xx}$ is globally Lipschitz in $y$ with Lipschitz constant $\Lip$.
	Hence $J(\alpha,\beta_n)_t < \infty$ by Taylor's theorem.
	Let $H(x,\Delta x,y):= G(x+\Delta x, y) - G(x, y) - G_x(x, y) \Delta x$ and denote by $\tilde J^{n,\pm}$ the increasing and decreasing components of $J(\alpha,\beta_n)-J(\alpha,\beta)$, respectively, i.e.\
	\[
		\tilde J^{n,+}_t := \sum_{\mathclap{\substack{u\le t \\ \tilde H(\dots) > 0 }}} \tilde H\bp{\alpha({u-}), \Delta\alpha(u), \beta_n(u), \beta(u)}
		\text{,} \ \ 
		\tilde J^{n,-}_t := \sum_{\mathclap{\substack{u\le t \\ \tilde H(\dots) < 0 }}} \tilde H\bp{\alpha({u-}), \Delta\alpha(u), \beta_n(u), \beta(u)},
	\]
	for $\tilde H(x,\Delta x, y, z):= H(x,\Delta x,y) - H(x,\Delta x,z)$. 
	Moreover, take any enumeration of the jump times of $\alpha$, $\{t_k \mid k\in\NN\} = \{t \mid \Delta\alpha(t) \ne 0\}$, and arbitrary $\varepsilon > 0$.
	Since $[\alpha]^d < \infty$, there exists $K \in \NN$ such that $\sum_{k > K} \abs{\Delta\alpha(t_k)}^2 < \varepsilon / (2 C \Lip)$.
	Moreover, we have $\abs{\tilde H(x,\Delta x, y,z)} \le \tfrac{1}{2} \abs{\Delta x}^2 \Lip \abs{y-z}$ and thus
	\[
		\abs{\tilde J^{n,\pm}_T} 
			\le \frac{\Lip}{2} \sum_{k=1}^\infty \abs{\Delta\alpha(t_k)}^2 \abs{\beta_n(t_k) - \beta(t_k)}
			< \frac{\varepsilon}{2} + \frac{\Lip}{2} \paren[\Big]{\max_{1\le k \le K} \abs{\beta_n(t_k) - \beta(t_k)}} \sum_{k=1}^K \abs{\Delta\alpha(t_k)}^2\,.
	\]
	By pointwise convergence $\beta_n(t_k) \to \beta(t_k)$ at all $t_k$, there exists $N \in \NN$ such that for all $k=1,\dots,K$ and $n \ge N$ we have $\abs{\beta_n(t_k) - \beta(t_k)} < \varepsilon/(L [\alpha]^d_T)$ and therefore $\abs{\tilde J^{n,\pm}_T} < \varepsilon$ for $n\ge N$.
	Hence $J^{n,\pm}_T \to 0$ as $n \to \infty$.
	
	Since $J^{n,\pm}$ are monotone and do not cross zero, we have $\sup_{0\le t\le T} \abs{\tilde J^{n,\pm}_t} = \abs{\tilde J^{n,\pm}_T}$ and therefore uniform convergence $\tilde J^{n,\pm} \to 0$ on $[0,T]$. 
	So in particular $J(\alpha,\beta_n)$ converges to $J(\alpha,\beta)$, uniformly on $[0,T]$.
\end{proof}

\label{sect:Marcus integral proofs}

\begin{proof}[Proof of \cref{lem:problem as Marcus integral}]
	Since $\assetsProcess$ is of finite variation, we have $\diff {[Z^j, Z^m]^c_t} = \diff [\baseS]^c_t$ for $j=m=2$, and $0$ otherwise.
	So the $\partial \marcusDriver_{\cdot,j} / \partial x_\ell$ terms in \cref{eq:defMarcusIntegral} simplify to 
	\begin{equation} \label{eq:MarcusDerivativeTerms}
		\frac{1}{2} \sum_{j,m=1}^3 \sum_{\ell = 1}^3 \int_0^t \frac{\partial \marcusDriver_{\cdot,j}}{\partial x_\ell}(X_{s-}) \marcusDriver_{\ell,m}(X_{s-}) \diff {[Z^j, Z^m]^c_s} = \paren{ 0,\ 0,\ 0 }^{tr}.
	\end{equation}
	Jumps of $Z$ are of the form 
	$\Delta Z_s = \paren[\big]{\Delta \assetsProcess_s,\ \Delta\baseS_s,\ 0}^{tr}$, so for $\xi(X):= \marcusDriver(X) \Delta Z_s$ we obtain
	\(
		\xi(X)
			= \paren[\big]{ -g(X^3, X^2) \Delta \assetsProcess_s,\ \Delta \assetsProcess_s,\ \Delta \baseS_s}^{tr},
	\)
	which yields the solutionto \eqref{eq:defPhiODE} as $y(u) = V_u = (V^1_u, V^2_u, V^3_u)^{tr} \in \RR^3$ with $V_0 = X_{s-}$\,,
	\begin{align*}
		V^2_u &= Y_{s-} + \int_0^u \Delta \assetsProcess_s \diff x = Y_{s-} + u \Delta \assetsProcess_s \,,
	\\	
		V^3_u &= \baseS_{s-} + \int_0^u \Delta \baseS_s \diff x = \baseS_{s-} + u \Delta\baseS_s \,,
	\allowdisplaybreaks
	\\	
		V^1_u &= L_{s-} - \int_0^u g(\baseS_{s-} + x \Delta\baseS_s, Y_{s-} + x \Delta \assetsProcess_s) \Delta \assetsProcess_s \diff x 
		\\	&= L_{s-} - \int_0^{u \Delta \assetsProcess_s} g(\baseS_{s-}, Y_{s-} + x) \diff x \,,
	\end{align*}
	since quasi-left continuity of $\baseS$ gives that a.s.\ $\Delta\baseS_s = 0$ whenever $\Delta\assetsProcess_s \ne 0$ (jumps of $\assetsProcess$ occur at predictable times).
	Thus the jump terms in \eqref{eq:defMarcusIntegral} become
	\begin{align} \label{eq:MarcusJumpTerms}
	\begin{split}
		\MoveEqLeft
		\varphi\paren{\marcusDriver(\cdot) \Delta Z_s, X_{s-}} - X_{s-} - \marcusDriver(X_{s-})\Delta Z_s 
	\\	
			&= \paren[\bigg]{ - \int_0^{\Delta \assetsProcess_s} g(\baseS_{s-}, Y_{s-} + x) \diff x + g(\baseS_{s-}, Y_{s-}) \Delta \assetsProcess_s ,\ 0,\ 0 }^{tr}.
	\end{split}
	\end{align}
	Furthermore, the Itô integral in \eqref{eq:defMarcusIntegral} reads
	\begin{align} \label{eq:MarcusItoTerm}
		\int_0^t \marcusDriver(X_{s-}) \diff Z_s = \begin{pmatrix}
			-\int_0^t g(\baseS_{s-}, Y_{s-}) \diff \assetsProcess_s
		\\	-\int_0^t h(Y_s) \diff \angles{M}_s + \assetsProcess_t - \assetsProcess_{0-}
		\\	\baseS_t - \baseS_{0-}
		\end{pmatrix}.
	\end{align}
	Summing up $X_{0-}$ and \cref{eq:MarcusItoTerm,eq:MarcusDerivativeTerms,eq:MarcusJumpTerms} yields the second and third components $Y_{0-} - \int_0^t h(Y_s) \diff s + \assetsProcess_t - \assetsProcess_{0-} = Y_t$ and $\baseS_{0-} + \baseS_t - \baseS_{0-} = \baseS_t$, respectively.
	To complete the proof, we note that for the first component we get
	\[
		L_{0-} - \int_0^t g(\baseS_{s-}, Y_{s-}) \diff \assetsProcess_s + \sum_{\substack{0\le s\le t \\ \Delta\assetsProcess_s \ne 0}} \paren[\Big]{ g(\baseS_{s-}, Y_{s-}) \Delta\assetsProcess_s - \int_0^{\Delta\assetsProcess_s} \!\! g(\baseS_{s-}, Y_{s-}+x) \diff x }
		= L_t\,.\ \qedhere
	\]
\end{proof}

\label{sect:finite stochastic horizon proofs}

The following proves the technical \cref{lem:admissible type F strategies} about admissibility of \typeFname{} strategies in \cref{ex: stochastic-finite-horizon}.
\begin{proof}[{Proof of \cref{lem:admissible type F strategies}}]
	By \cite[Ch.~2, Sect.~2, eq.~(2.0.2) on p.~295]{BorodinSalminen02}, the law of the hitting time $H_z$ of level $z$ by a Brownian motion with drift $\mu$ starting in $x$ is for $t \in (0,\infty)$ given by $\PP_x[ H_z \in \!\diff t ] = h^\mu(t,z-x) \diff t$ with $h^\mu(t,x) := \frac{\abs{x}}{\sqrt{2\pi} t^{3/2}} \exp\bp{ - \frac{(x-\mu t)^2}{2t} }$ and $\PP_x[H_z = \infty] = 1 - \exp\bp{ \mu (z-x) - \abs{\mu} \cdot \abs{z-x} }$.
	With $\mu = \beta \tilde\Upsilon / \rho$, $x = (\assetsProcess_{0-} + \tilde\Upsilon - Y_{0-})/\rho$ and $z = (\tilde\Upsilon - \Upsilon)/\rho$ we obtain the stated terms for $\EE[\tau] = \EE_x[H_z]$.
	
	Now, let $\tilde\Upsilon$, $\Upsilon$ be such that $\EE[\tau] \le \maxET$.
	Independence of $\tau$ and $\baseS$ gives $\EE[\baseS_\tau] = \baseS_0$ and $\EE[\tau \baseS_\tau] < \infty$.
	We have $\int_0^\tau \baseS_t f(Y_t) \diff B_t = f(\tilde\Upsilon) M_\tau$ for $M_T:= \int_0^{T \wedge \tau} \baseS_t \diff B_t$
	and $\int_0^\tau \baseS_t F(Y_t) \diff W_t = F(\tilde\Upsilon) \sigma^{-2} \baseS_\tau$.
	Note that $\brackets{M}_\tau = \sigma^{-2} \brackets{\baseS}_\tau $.
	We will show that $M$, $\baseS_{\cdot\wedge \tau}$ and $(\baseS B)_{\cdot\wedge \tau}$ are in $\scH^1$ and hence UI martingales.
	By Burkholder-Davis-Gundy \cite[Thm.~IV.4.48]{Protter04}, there exists $C > 0$ such that $\EE\brackets[\big]{ [\baseS]_\tau^{1/2} } \le C \EE[\sup_{u \le \tau} \abs{\baseS_u}] = C\EE[\exp(\sigma X_\tau)]$ with $X_t:= \sup_{u \le t} (W_u - \frac{\sigma}{2}u)$.
	Using $\{X_t > z\} = \{ H_z < t \}$ for $z,t \ge 0$ with starting point $X_0 = 0$ and drift $\mu = -\sigma/2$ we first obtain
	\begin{align*}
		\EE[\exp(\sigma X_t)] 
			&= \int_{[0,\infty]} e^{\sigma x} \PP[X_t \in \!\diff x]
			= \int_{[0,\infty]} e^{\sigma x} \diff \bp{1 - \PP[X_t > x]}_x
		\\	&= -\int_{[0,\infty]} e^{\sigma x} \diff \bp{\PP[X_t > x]}_x
			= -\int_{[0,\infty]} e^{\sigma x} \diff \bp{\PP[H_x < t]}_x\,.
	\end{align*}
	Since $\PP[H_\infty < t] = 0$ we can approximate the Riemann-Stieltjes integral and apply integration by parts twice to get
	\begin{align*}
		&\hspace{-\mathindent}\hspace{1em}
		\EE[\exp(\sigma X_t)] 
			= - \lim_{\varepsilon \searrow 0} \int_0^t \int_\varepsilon^{1/\varepsilon} e^{\sigma x} h^{-\sigma/2}_x(u, x) \diff x \diff u
			= -\int_0^t \int_0^\infty e^{\sigma x} h^{-\sigma/2}_x(u, x) \diff x \diff u
	\allowdisplaybreaks	
	\\	&\hspace{-\mathindent}\text{with}\quad
		h^{-\sigma/2}_x(t,x) 
			= \frac{\diff}{\diff x} h^{-\sigma/2}(t,x) 
			= -\frac{x^2 - t + \frac{\sigma}{2} x t}{\sqrt{2\pi} t^{5/2}} \exp\paren[\Big]{ -\frac{(x+\frac{\sigma}{2}t)^2}{2t} }\,.
	\end{align*}
	So we have $e^{\sigma x} h^{-\sigma/2}_x(t,x) = h^{\sigma/2}_x(t,x) - \sigma h^{\sigma/2}(t,x)$.
	The contribution from the first summand of the integrand $h^{\sigma/2}_x(t,x) - \sigma h^{\sigma/2}(t,x)$ is zero, since $h^{\sigma/2}(t,x) \to 0$ for $x \to \infty$ and for $x \to 0$.
	Hence, $\EE[ \exp(\sigma X_t) ]$ equals
	\begin{align*}
		&\sigma \int_0^t \int_0^\infty h^{\sigma/2}(u,x) \diff x \diff u
= \sigma \int_0^t \paren[\bigg]{ \frac{\exp\bp{ -\frac{\sigma^2}{8}u }}{\sqrt{2\pi u}} - \frac{\sigma}{2} + \frac{\sigma}{2}\varphi\paren[\big]{\tfrac{\sigma}{2}\sqrt{u}} } \diff u
	\allowdisplaybreaks
	\\		&= 2 \varphi\bp{ \tfrac{\sigma}{2}\sqrt{t} } - 1 + \tfrac{\sigma^2}{2}t \varphi\bp{ \tfrac{\sigma}{2}\sqrt{t} } - \tfrac{\sigma^2}{2}t + \frac{\sigma \sqrt{t}}{\sqrt{2\pi}} \exp\bp{ -\tfrac{\sigma^2}{8}t }
		\le 1 + \frac{\sigma\sqrt{t}}{\sqrt{2\pi}}
	\,,
	\end{align*}
	where $\varphi(x) = \int_{-\infty}^x e^{-z^2/2} \diff z / \sqrt{2\pi}$.
	So by independence of $X$ and $\tau$
	\begin{align*}
		\hspace{-1em}
		\EE[\exp(\sigma X_\tau)] 
			&= \EE[ (t \mapsto \EE[ e^{\sigma X_t} ])(\tau) ] 
			\le \EE\brackets[\Big]{ 1 + \frac{\sigma}{\sqrt{2\pi}}\sqrt{\tau} } 
			\le 1 + \frac{\sigma}{\sqrt{2\pi}}(1+\EE[\tau]) 
			< \infty\,.
	\end{align*}	
	Moreover, $\brackets{\baseS B}_\tau = \tau \brackets{\baseS}_\tau$ by independence of $\baseS$ and $B$,
	so we can bound $\EE\brackets[\big]{ \brackets{\baseS B}_\tau^{1/2} } $ by
	\(
		\EE\brackets[\big]{ \sqrt{\tau}\brackets{\baseS}_\tau^{1/2} }
		\!=\! \EE\brackets[\Big]{ {\sqrt{t} \EE\brackets[\big]{ \brackets{\baseS}_t^{1/2} } }\Bigr|_{t=\tau} }
		\!\le\! C \EE\brackets[\Big]{ { \sqrt{t} \EE\brackets{ \exp(\sigma X_t) } }\Bigr|_{t=\tau} }
		\!\le\! C\EE\brackets{ \sqrt{\tau} + \tfrac{\sigma}{\sqrt{2\pi}} \tau }
		< \infty.
	\)
	Thus, $(\baseS B)_{\cdot \wedge \tau}$ is in $\scH^1$ and hence a UI martingale.
	
	Finally, $\bp{L_\tau(\assetsProcess)}^- \in L^1(\PP)$ follows from  $\int_0^\tau \baseS_t g(Y^\assetsProcess_{t-}) \diff t = g(\tilde \Upsilon)\int_0^\tau \baseS_t \diff t $, which is integrable by optional projection \cite[Thm.~VI.57]{DellacherieMeyer82bookB} since $\EE[\tau \baseS_\tau] < \infty$, and integrability of $ \baseS_\tau F(Y^\assetsProcess_\tau) = \baseS_\tau F(\Upsilon)$.
\end{proof}

\footnotesize{

}

\end{document}